%% file: paper.tex
\documentclass[a4paper,english,autoref]{lipics-v2019}

\usepackage[utf8]{inputenc}
\usepackage{amsmath}
\usepackage{amsfonts}
\usepackage[title]{appendix}
\usepackage{comment}
\usepackage{hyperref}
\usepackage{mathtools}
\usepackage{tikz}

\usetikzlibrary{arrows,automata,shapes,positioning}

\nolinenumbers

\allowdisplaybreaks

\newcommand{\true}{\mathit{true}}
\newcommand{\false}{\mathit{false}}
\newcommand{\dom}{\textit{dom}}

\DeclareFontFamily{OT1}{pzc}{}
\DeclareFontShape{OT1}{pzc}{m}{it}{<-> s * [1.10] pzcmi7t}{}
\DeclareMathAlphabet{\mathpzc}{OT1}{pzc}{m}{it}

\EventEditors{Igor Konnov and Laura Kov\'{a}cs}
\EventNoEds{2}
\EventLongTitle{31st International Conference on Concurrency Theory (CONCUR 2020)}
\EventShortTitle{CONCUR 2020}
\EventAcronym{CONCUR}
\EventYear{2020}
\EventDate{September 1--4, 2020}
\EventLocation{Vienna, Austria}
\EventLogo{}
\SeriesVolume{2017}
\ArticleNo{35}

\title{Propositional Dynamic Logic for Hyperproperties}
\author{Jens Oliver Gutsfeld}{Institut für Informatik, Westfälische Wilhelms-Universität Münster, Germany}{jens.gutsfeld@uni-muenster.de}{}{}
\author{Markus Müller-Olm}{Institut für Informatik, Westfälische Wilhelms-Universität Münster, Germany}{markus.mueller-olm@uni-muenster.de}{}{}
\author{Christoph Ohrem}{Institut für Informatik, Westfälische Wilhelms-Universität Münster, Germany}{christoph.ohrem@uni-muenster.de}{}{}
\authorrunning{J. O. Gutsfeld \and M. Müller-Olm \and C. Ohrem}
\Copyright{Jens Oliver Gutsfeld \and Markus Müller-Olm \and Christoph Ohrem}


\acknowledgements{We thank the reviewers for their helpful comments.}
\funding{This work was partially funded by DFG project Model-Checking of Navigation Logics (MoNaLog) (MU 1508/3).}
\relatedversion{The conference version of this paper will appear in the proceedings of CONCUR 2020.}

\begin{document}
	\maketitle
	\input{sections/introduction/abstract}
	\input{sections/introduction/introduction}

	\input{sections/introduction/preliminaries}

	\input{sections/main/hyperpdl}
	\input{sections/main/modelchecking}

	\input{sections/main/satisfiability}
	\input{sections/main/expressivity}

	\input{sections/conclusion/conclusion}
	\bibliographystyle{plainurl}
	\bibliography{sections/conclusion/citations}
	\input{sections/conclusion/appendix}

\end{document}

%% file: sections/introduction/abstract.tex
\begin{abstract}
Information security properties of reactive systems like non-interference often require relating different executions of the system to each other and following them simultaneously.
Such \textit{hyperproperties} can also be useful in other contexts, e.g., when analysing properties of distributed systems like linearizability.
Since common logics like LTL, CTL, or the modal $\mu$-calculus cannot express hyperproperties, the hyperlogics HyperLTL and HyperCTL$^*$ were developed to cure this defect.
However, these logics are not able to express arbitrary $\omega$-regular properties.
In this paper, we introduce HyperPDL-$\Delta$, an adaptation of the Propositional Dynamic Logic of Fischer and Ladner for hyperproperties, in order to remove this limitation.
Using an elegant automata-theoretic framework, we show that HyperPDL-$\Delta$ model checking is asymptotically not more expensive than HyperCTL$^*$ model checking, despite its vastly increased expressive power.
We further investigate fragments of HyperPDL-$\Delta$ with regard to satisfiability checking.
\end{abstract}

%% file: sections/introduction/introduction.tex
\section{Introduction}

Temporal logics like LTL, CTL or CTL$^*$ have been used successfully in verification.
These logics consider paths of a structure (in linear time logics) or  paths and their possible extensions (in branching time logics).
Notably, since they cannot refer to multiple paths at once, they cannot express \textit{hyperproperties} that relate multiple paths to each other.
Examples of hyperproperties include information security properties like non-interference \cite{Clarkson2010} or properties of distributed systems like linearizablity \cite{Bonakdarpour2018}.
In order to develop a dedicated logic for these properties, Clarkson et. al. \cite{Clarkson2014,Finkbeiner2015} introduced HyperLTL and HyperCTL$^*$, which extend LTL and CTL$^*$ by path variables. 
However, just like LTL and CTL \cite{Wolper1983}, they cannot express arbitrary $\omega$-regular properties of traces \cite{Rabe2016}, a desirable property of specification logics \cite{Armoni2002,Kupferman2001}.
Logics like Propositional Dynamic Logic (PDL) \cite{Fischer1979,Lange2006} and Linear Dynamic Logic (LDL)
\cite{DeGiacomo2013,Faymonville2017} are able to do so for single traces. 
As we seek to extend this ability to hyperproperties,  we introduce a variant of PDL for hyperproperties called HyperPDL-$\Delta$ in this paper.
HyperPDL-$\Delta$ properly extends logics like HyperLTL, HyperCTL$^*$, LDL and PDL-$\Delta$ and can express all $\omega$-regular properties over hypertraces in a handy formalism based on regular expressions over programs.
We develop a model checking algorithm for HyperPDL-$\Delta$ inspired by one for HyperCTL$^*$ \cite{Finkbeiner2015} and show that the model checking problem for HyperPDL-$\Delta$ is decidable at no higher asymptotic cost than the corresponding problem for HyperCTL$^*$ despite the vastly increased expressive power.
Our algorithm non-trivially differs from the algorithm in \cite{Finkbeiner2015} in two ways: first of all, it handles more general, regular modalities that subsume the modalities of HyperCTL$^*$ and require different constructions.
Then, we use a different notion of alternation depth (called \textit{criticality}) which conservatively extends their notion, but requires handling structurally different operators and regular expressions.
We also show that for fragments of HyperPDL-$\Delta$ similar to the fragments of HyperLTL considered in \cite{Finkbeiner2016}, the satisfiability problem is decidable. 

This paper is structured as follows: \autoref{section:preliminaries} introduces Kripke Transition Systems and (alternating) Büchi automata.
In \autoref{section:HyperPDL}, we define our new logic HyperPDL-$\Delta$ and describe properties expressible in it.
Afterwards, in \autoref{Section:ModelChecking}, we outline a model checking algorithm for HyperPDL-$\Delta$ and show that it is asymptotically optimal by providing a precise complexity classification. In \autoref{section:satisfiability}, we consider fragments of HyperPDL-$\Delta$ for which the satisfiability problem is decidable. Then, in \autoref{section:expressivity}, we show that HyperPDL-$\Delta$ can express all $\omega$-regular properties over sets of traces and compare it to existing hyperlogics with regard to expressivity. Finally, in \autoref{section:conclusion}, we provide a summary of this paper.
The appendices contain two constructions only sketched in the main body of this paper.
Due to lack of space, some proofs can be found in the appendix of this extended version only. 

\textbf{Related Work.} 
Hyperproperties were systematically analysed in \cite{Clarkson2010} and dedicated temporal logics for hyperproperties, HyperLTL and HyperCTL$^*$, were introduced in \cite{Clarkson2014}. 
An overview of temporal hyperlogics and discussion of their expressive power can be found in \cite{Coenen2019}.
Efficient model checking algorithms for these logics were introduced in \cite{Finkbeiner2015} by Finkbeiner et al. and our model checking algorithm for HyperPDL-$\Delta$ builds on ideas from their construction.
Our satisfiability algorithm, on the other hand, is inspired by the corresponding algorithm for HyperLTL \cite{Finkbeiner2016}.
Recently, Bonakdapour et. al.  proposed \emph{regular hyperlanguages} and a corresponding automata model \cite{Bonakdarpour2020}.
In contrast to our work, their model is concerned with hyperproperties over finite instead of infinite words and does not concern branching-time properties.
Moreover, they study automata-theoretic questions while our focus here is on verification of hyperproperties specified by logical means.
A different line of research for properties involving multiple traces at once is given by \textit{epistemic temporal logics} \cite{Halpern2004}.
An attempt to unify epistemic temporal logics and hyperlogics is given by Bozzelli et. al in \cite{Bozzelli2015} via a variant of HyperCTL$^*$ with past modalities.
PDL was originally introduced in \cite{Fischer1979} by Fischer and Ladner and has 
been extended in multiple ways \cite{Lange2006}.
There are several attempts to extend temporal logic by regular properties: a variant of PDL for linear time properties of finite traces, LDL$_f$, was introduced in \cite{DeGiacomo2013} and was extended upon for the infinite setting in \cite{Faymonville2017} by introducing parametrised operators.
Other regular extensions of temporal logics were studied e.g. by Wolper \cite{Wolper1983} or by Kupferman et. al \cite{Kupferman2001}. 
However, all these extensions do not concern hyperproperties.

%% file: sections/introduction/preliminaries.tex
\section{Preliminaries}\label{section:preliminaries}


Let $AP$ be a finite set of atomic propositions and $\Sigma$ a finite set of atomic programs.
A \textit{Kripke Transition System} (KTS) is a tuple $\mathcal{T} = (S, s_0, \{\delta_{\sigma}\mid\sigma \in \Sigma\}, L)$ where $S$ is a finite set of states, $s_0 \in S$ is an initial state, $\delta_\sigma \subseteq S \times S$ is a transition relation for each $\sigma \in \Sigma$ and $L : S \to 2^{AP}$ is a labeling function.
We assume that there are no states without outgoing edges, that is for each $s \in S$, there is $s' \in S$ with $(s,s') \in \delta_{\sigma}$ for some $\sigma \in \Sigma$.
A KTS is a combination of a Kripke Structure and a labeled transition system (LTS) where a Kripke Structure $K$ is a KTS for $|\Sigma| = 1$ and an LTS $T$ is a KTS with the labelling function $s \mapsto \emptyset$ for all $s \in S$.
A path in a KTS $\mathcal{T}$ is an infinite alternating sequence $s_0 \sigma_0 s_1 \sigma_1 ... \in (S\Sigma)^{\omega}$ where $s_0$ is the initial state of $\mathcal{T}$ and $(s_{i}, s_{i+1})\in \delta_{\sigma_i}$ for all $i \geq 0$.
We denote by $\mathit{Paths}(\mathcal{T},s)$ the set of paths in $\mathcal{T}$ starting in $s$ and by $\mathit{Paths}^{*}(\mathcal{T},s)$ the set of corresponding path suffixes $\{p[i,\infty] \mid p \in \mathit{Paths}(\mathcal{T},s), i \in \mathbb{N}_0\}$ where $p[i,\infty]$ is the path suffix of $p$ starting at index $i$.
A trace is an alternating infinite sequence $t \in (2^{AP} \Sigma)^{\omega}$. For a path $\pi = s_0 \sigma_0 s_1 \sigma_1 ...$, the induced trace is given by $L(s_0) \sigma_0 L(s_1) \sigma_1 ...$. For a KTS $\mathcal{T}$ and a state $s \in S$, we write $\mathit{Traces}(\mathcal{T}, s)$ to denote the traces induced by paths of $\mathcal{T}$ starting in $s$. 

An alternating Büchi automaton (ABA) is a tuple $\mathcal{A} = (Q,q_0,\Sigma,\rho,F)$ where $Q$ is a finite set of states, $q_0 \in Q$ is an initial state, $\Sigma$ is a finite alphabet, $\rho : Q \times \Sigma \to \mathbb{B} ^{+}(Q)$ is a transition function mapping each pair of state and input symbol to a non-empty positive boolean combination of successor states and $F \subseteq Q$ is a set of accepting states.
We assume that every ABA has two distinct states $\true \in F$ and $\false \in Q \setminus F$ with $\rho(\true,\sigma) = true$ and $\rho(\false,\sigma) = false$ for all $\sigma \in \Sigma$.
Thus, all maximal paths in an ABA are infinite.
A tree $T$ is a subset of $\mathbb{N}^{*}$ such that for every node $t \in \mathbb{N}^{*}$ and every positive integer $n \in \mathbb{N}$: $t \cdot n \in T$ implies (i) $t \in T$ (we then call $t \cdot n$ a child of $t$), and (ii) for every $0 < m < n$, $t \cdot m \in T$. 
We assume every node has at least one child.
A path in a tree $T$ is a sequence of nodes $t_0 t_1 ...$ such that $t_0 = \varepsilon$ and $t_{i+1}$ is a child of $t_i$ for all $i \in \mathbb{N}_{0}$.
A run of an ABA $\mathcal{A}$ on an infinite word $w \in \Sigma^{\omega}$ is defined as a $Q$-labeled tree $(T,r)$ where $r: T \to Q$ is a labelling function such that $r(\varepsilon) = q_{0}$ and for every node $t \in T$ with children $t_1,...,t_k$, we have $1 \leq k \leq |Q|$ and the valuation assigning true to the states $r(t_1),...,r(t_k)$ and false to all other states satisfies $\rho(r(t),w(|t|))$.
A run $(T,r)$ is an accepting run iff for every path $t_1 t_2 ...$ in $T$, there are infinitely many $i$ with $r(t_i) \in F$.
A word $w$ is accepted by $\mathcal{A}$ iff there is an accepting run of $\mathcal{A}$ on $w$.
The set of infinite words accepted by $\mathcal{A}$ is denoted by $\mathcal{L}(\mathcal{A})$. A nondeterministic Büchi automaton is an ABA in which every transition rule consists only of disjunctions.

We will make use of two well-known theorems about ABA:

\input{math/theorems/dealternation}

\input{math/theorems/complementation}

%% file: math/theorems/dealternation.tex
\begin{proposition}[\cite{Miyano1984}]\label{dealternation}
	For every ABA $\mathcal{A}$ with $n$ states, there is a nondeterministic Büchi automaton $\text{MH}(\mathcal{A})$ with $2^{\mathcal{O}(n)}$ states that accepts the same language.
\end{proposition}

%% file: math/theorems/complementation.tex
\begin{proposition}[\cite{Miyano1984},\cite{Kupferman2001a}]\label{complementation}
	For every ABA $\mathcal{A}$ with $n$ states, there is an ABA $\overline{\mathcal{A}}$ with $\mathcal{O}(n^{2})$ states that accepts the complement language, i.e., $\mathcal{L}(\overline{\mathcal{A}})= \overline{\mathcal{L}(\mathcal{A})}$.
\end{proposition}

%% file: sections/main/hyperpdl.tex
\section{Propositional Dynamic Logic for Hyperproperties}\label{section:HyperPDL}
In this section, we define our new logic, HyperPDL-$\Delta$.
Structurally, it consists of formulas $\varphi$ referring to state labels and programs $\alpha$ referring to transition labels.
We use the syntax of HyperCTL* as a basis for formulas but replace the modalities $\bigcirc, \mathcal{U}$ and $\mathcal{R}$ by PDL-like expressions $\langle \alpha \rangle \varphi, [\alpha] \varphi$ and $\Delta \alpha$ constructed from programs.
These programs $\alpha$ are regular expressions over tuples of atomic programs $\tau$ capturing the transition behaviour on the considered paths.
Additionally, we allow test-operators $\varphi ?$ in $\alpha$ in order to enable constructions like conditional branching.
\input{math/defs/syntax}

We call a HyperPDL-$\Delta$ formula $\varphi$ \textit{closed} iff all occurrences of path variables $\pi$ in $\varphi$ (as indices of atomic propositions or in atomic programs) are bound by a quantifier.
In this paper, we only consider closed formulas $\varphi$.
In programs $\alpha$, each component of tuples $\tau \in (\Sigma \cup \{\cdot\})^n$ corresponds to one of the path variables bound by a quantifier $\alpha$ is in scope of.
We assume that a quantifier that is in scope of $i-1$ other quantifiers quantifies path variable $\pi_i$.

Connectives inherited from HyperCTL$^*$ are interpreted analogously:
quantifiers $\exists$ and $\forall$ should be read as ``along some path''  and ``along all paths''.
Using different path variables $\pi$ enables us to refer to multiple paths at the same time.
For example, with $\forall \pi_1. \exists \pi_2. \exists \pi_3. \varphi$, one can express that for all paths $\pi_1$, there are paths $\pi_2$ and $\pi_3$ such that $\varphi$ holds along these three paths.
Boolean connectives are defined in the usual way.
Atomic propositions $a \in AP$ express information about a state and have to be indexed by a path variable $\pi$ to express on which path we expect $a$ to hold.

Intuitively, a program $\alpha$ explores all paths it is in scope of synchronously and thus allows us to pose a regular constraint on the sequence of atomic programs visited on path prefixes of equal length. 
In addition, properties of infinite suffixes can be required at certain points during the exploration using tests $\varphi ?$.
In formulas $\varphi$, atomic propositions $a_{\pi}$ on single paths suffice to relate behavior on different paths because boolean connectives are available.
Referring to occurences of atomic programs on single paths in programs $\alpha$ in a similar way, however, would reduce expressivity since neither negation nor conjunction are present in $\alpha$. 
As a remedy, we use tuples $\tau \in (\Sigma \cup \{\cdot\})^n$ to refer to atomic programs on paths $\pi_1, ..., \pi_n$, where $\sigma$ in position $i$ means that on path $\pi_i$, we expect a use of $\sigma$ to reach the next state.
A wildcard symbol $\cdot$ expresses that any atomic program is allowed on the corresponding path.

The constructs using $\alpha$ can be interpreted as follows:
the diamond operator $\langle \alpha \rangle \varphi$ means that $\alpha$ matches a prefix of the current paths after which $\varphi$ holds.
The box operator $[\alpha] \varphi$ is the dual of $\langle \alpha \rangle \varphi$, meaning $\varphi$ holds at the end of all prefixes matching $\alpha$.
The last construct $\Delta \alpha$ is of a different kind and expresses $\omega$-regular rather than regular properties.
It says that $\alpha$ occurs repeatedly, i.e., the currently quantified paths can be divided into infinitely many segments matching $\alpha$.
$\Delta \alpha$ expresses a variant of a Büchi condition.
Instead of moving from accepting states to accepting states in a Büchi automaton, one moves from initial states to accepting states repeatedly.

Using our logic, common hyperproperties can be expressed easily and intuitively.
Let us consider two examples: the first, \textit{observational determinism} \cite{Clarkson2010}, states that if two executions of a system receive equal low security inputs, they are indistinguishable for a low security observer all the time.
It can be expressed by $\forall \pi_1 . \forall \pi_2 . (\bigwedge_{a \in L} (a_{\pi_1} \leftrightarrow a_{\pi_2})) \rightarrow [\bullet^*] \bigwedge_{a \in L} (a_{\pi_1} \leftrightarrow a_{\pi_2})$, where low security observable behaviour is modelled by the atomic propositions in $L$.
Here, we use the common boolean abbreviations $\rightarrow$ for implication and $\leftrightarrow$ for equivalence as well as $\bullet$ as an abbreviation the program $\tau = (\cdot,...,\cdot)$.
The second example, \textit{generalized noninterference} \cite{Clarkson2010}, states that high security injections do not interfere with low security observable behaviour. 
It can be expressed by stating that for all pairs of executions $\pi_1,\pi_2$ there is a third execution $\pi_3$ agreeing with $\pi_1$ on high security injections and is indistinguishable from $\pi_2$ for a low security observer: $\forall \pi_1 . \forall \pi_2 . \exists \pi_3 . [\bullet^*] \bigwedge_{a \in H} (a_{\pi_1} \leftrightarrow a_{\pi_3}) \land \bigwedge_{a \in L} (a_{\pi_2} \leftrightarrow a_{\pi_3})$.
Since these hyperproperties can already be expressed in HyperLTL \cite{Clarkson2014}, which is subsumed by our logic by encoding $\varphi_1 \mathcal{U} \varphi_2$ formulas with $\langle (\varphi_1 ? \cdot \bullet)^{*} \rangle \varphi_2$, it is no surprise, that they can be expressed in HyperPDL-$\Delta$ as well.

The ability to express arbitrary $\omega$-regular properties however, allows a much more fine-grained analysis of a system than HyperLTL.
For example, by replacing the program $\bullet^*$ with $(\bullet \cdot \bullet)^* \cdot \bullet \cdot (\sigma_1,\sigma_1)$ in the observational determinism formula, we can restrict the requirement on low security outputs to only apply for every other state with the additional constraint that some specific program $\sigma_1$ was last executed in both $\pi_1$ and $\pi_2$.
As was argued in \cite{Armoni2002} for linear time specification logics, the ability to express $\omega$-regular properties can indeed become a practical issue when in an assume-guarantee setting detailed information about the behaviour of the context has to be taken into account in order to prove properties of interest.
This indicates that availability of $\omega$-regular properties is not just a theoretical issue in specification logics.

While there are hyperlogics with the ability to express all $\omega$-regular languages \cite{Coenen2019}, these lack properties desirable for verification:
HyperQPTL obtains the ability to express $\omega$-regular properties from the addition of propositional quantification, which complicates its use for specification purposes since the user has to keep track of heterogeneous types of quantifiers.
The simple property that all executions $\pi_1$ and $\pi_2$ agree on propositions from a set $P$ in every other state for example is expressed by the HyperQPTL formula $\forall \pi_1. \forall \pi_2.\exists t : t \land \square (\bigcirc t \leftrightarrow \lnot t) \land \square(t \rightarrow \bigwedge_{a \in P}a_{\pi_1} \leftrightarrow a_{\pi_2})$ \cite{Kesten2002}.
The specification of this property in HyperPDL-$\Delta$ is much more direct:  $\forall \pi_1.\forall \pi_2. [(\bullet \cdot \bullet)^*] \bigwedge_{a \in P} a_{\pi_1} \leftrightarrow a_{\pi_2}$.
This example also illustrates another problem: due to the additional quantifier alternation, the only known model checking algorithm for HyperQPTL \cite{Rabe2016} is exponentially more expensive than that of HyperPDL-$\Delta$ for such formulas.
S1S[E] on the other hand, while being even more expressive than HyperPDL-$\Delta$, has an undecidable model checking problem \cite{Coenen2019}.

Before formally defining our logic's semantics, we introduce some notation.
We call a partial function $\Pi: N \leadsto \mathit{Paths}^{*}(\mathcal{T},s_0)$ with $\dom(\Pi) = \{\epsilon,\pi_1,...,\pi_n\}$ a path assignment and denote by $\textit{PA}$ the set of all path assignments.
In the context of a subformula $\varphi$, $\dom(\Pi)$ contains exactly the variables it is in scope of as well as $\epsilon$.
The path variable $\epsilon$ refers to the most recently \textit{assigned} path in a path assignment and is used to ensure that paths induced by quantifiers branch from the most recently quantified path.
We use $\{\epsilon \to p\}$ for a path $p$ to denote the path assignment $\Pi$ with $\dom(\Pi) = \{\epsilon\}$ and $\Pi(\epsilon) = p$.
We introduce $\Pi[i,\infty]$ as a notation to manipulate path assignments $\Pi$ such that $\Pi[i,\infty](\pi) = \Pi(\pi)[i,\infty]$ holds for all $\pi \in \dom(\Pi)$.
Also, $\Pi[\pi_i \to p]$ is a notation for a path assignment $\Pi'$ where $\Pi'(\pi_i) = p$ and $\Pi'(\pi_j) = \Pi(\pi_j)$ for all $j \neq i$.
As a convention, we do not count $\epsilon$ when determining $|\dom(\Pi)|$.
For a tuple $\tau = (\sigma_1,...,\sigma_n)$, we write $\tau|_i$ to refer to $\sigma_i$.

We write $\Pi \models_{\mathcal{T}} \varphi$ to denote that in the context of a KTS $\mathcal{T}$, a path assignment $\Pi$ fulfills a formula $\varphi$.
We also write $(\Pi,i,k) \in R(\alpha)$ for a path assignment $\Pi$ and two even numbers $i \leq k$ to denote that the transition labels on the paths in $\Pi$ between $i$ and $k$ match $\alpha$.
Formally, we define these two relations as follows:

\input{math/defs/semantics}

A KTS $\mathcal{T}$ satisfies a formula $\varphi$, denoted by $\mathcal{T} \models \varphi$, iff $\{\epsilon \to p\} \models_{\mathcal{T}} \varphi$ holds for an arbitrary $p \in \mathit{Paths}(\mathcal{T},s_0)$.
Note that the choice of $p$ ensures that the outermost quantified paths in a formula always branch from the starting state of $\mathcal{K}$, i.e. $s_0$.

%% file: math/defs/syntax.tex
\begin{definition}[Syntax of HyperPDL-$\Delta$]
	Let $N = \{\epsilon, \pi_1, \pi_2 \ldots \}$ be a set of path variables with a special path variable $\epsilon \in N$. A formula $\varphi$ is a HyperPDL-$\Delta$ formula if it is built from the following context-free grammar:
	\begin{align*}
		\varphi ::= &\; \exists \pi . \varphi \mid \forall \pi . \varphi \mid a_{\pi} \mid \lnot \varphi \mid \varphi \land \varphi \mid \varphi \lor \varphi \mid \langle \alpha \rangle \varphi \mid [ \alpha ]\varphi \mid \Delta \alpha \\
		\alpha ::= &\; \tau \mid \varepsilon \mid \alpha + \alpha \mid \alpha \cdot \alpha \mid \alpha^{*} \mid \varphi ?
	\end{align*}
	where $\pi \in N \setminus \{\epsilon\}$, $a \in AP$ and $\tau \in (\Sigma \cup \{\cdot\})^n$ for the number $n > 0$ of path quantifiers that $\alpha$ is in scope of.
	The constructs $\langle \alpha \rangle \varphi, [ \alpha ]\varphi$ and $\Delta \alpha$ are only allowed in scope of at least one quantifier.
\end{definition}

%% file: math/defs/semantics.tex
\begin{definition}[Semantics of HyperPDL-$\Delta$]
	Given a KTS $\mathcal{T} = (S,s_0,\{\delta_{\sigma} \mid \sigma \in \Sigma\},L)$, we inductively define both satisfaction of formulas $\varphi$ and programs $\alpha$ on path assignments $\Pi$. 
	\begin{align*}
		&\Pi \models_{\mathcal{T}} \exists \pi. \varphi &\text{iff } & \text{there is } p \in \mathit{Paths}(\mathcal{T},\Pi(\epsilon)(0)) \text{ s.t. } \Pi[\pi \to p, \epsilon \to p] \models_{\mathcal{T}} \varphi \\
		&\Pi \models_{\mathcal{T}} \forall \pi. \varphi &\text{iff } & \text{for all } p \in \mathit{Paths}(\mathcal{T},\Pi(\epsilon)(0)) : \Pi[\pi \to p, \epsilon \to p] \models_{\mathcal{T}} \varphi \\
		&\Pi \models_{\mathcal{T}} a_{\pi} &\text{iff } & a \in L(\Pi(\pi)(0)) \\
		&\Pi \models_{\mathcal{T}} \lnot \varphi &\text{iff }  &\Pi \not\models_{\mathcal{T}} \varphi \\
		&\Pi \models_{\mathcal{T}} \varphi_{1} \land \varphi_{2} &\text{iff } & \Pi\models_{\mathcal{T}} \varphi_{1} \text{ and } \Pi\models_{\mathcal{T}} \varphi_{2} \\
		&\Pi \models_{\mathcal{T}} \varphi_{1} \lor \varphi_{2} &\text{iff } & \Pi\models_{\mathcal{T}} \varphi_{1} \text{ or } \Pi\models_{\mathcal{T}} \varphi_{2} \\
		&\Pi \models_{\mathcal{T}} \langle \alpha \rangle \varphi &\text{iff } & \text{there is } i \geq 0 \text{ s.t. } \Pi[i,\infty] \models_{\mathcal{T}} \varphi \text{ and } (\Pi,0,i) \in R(\alpha) \\
		&\Pi \models_{\mathcal{T}} [ \alpha ] \varphi &\text{iff } & \text{for all } i \geq 0 \text{ with } (\Pi,0,i) \in R(\alpha) : \Pi[i,\infty] \models_{\mathcal{T}} \varphi \\
		&\Pi \models_{\mathcal{T}} \Delta \alpha &\text{iff }  & \text{there are } 0 = k_1 \leq k_2 \leq ... \text{ s.t. for all } i \geq 1 : \\
		& & &(\Pi,k_{i},k_{i+1}) \in R(\alpha)
		\end{align*}
		\begin{align*}
		&(\Pi,i,k) \in R(\tau) &\text{iff }  & k = i+2 \text{ and for all } 1 \leq l \leq |\dom(\Pi)|: \\
		& & &\tau|_l = \cdot \text{ or }\Pi(\pi_l)(i+1) = \tau|_l \\
		&(\Pi,i,k) \in R(\varepsilon) &\text{iff }  & i = k \\
		&(\Pi,i,k) \in R(\alpha_{1}+\alpha_{2}) &\text{iff }  & (\Pi,i,k) \in R(\alpha_{1}) \text{ or } (\Pi,i,k) \in R(\alpha_{2}) \\
		&(\Pi,i,k) \in R(\alpha_{1} \cdot \alpha_{2}) &\text{iff }  & \text{there is } j \text{ s.t. } i \leq j \leq k, (\Pi,i,j) \in R(\alpha_{1})\text{ and } \\
		& & & (\Pi,j,k) \in R(\alpha_{2})\\
		&(\Pi,i,k) \in R(\alpha^{*}) &\text{iff }  & \text{there are } l \geq 0,i = j_{0} \leq j_{1} \leq ... \leq j_{l} = k \text{ s.t.} \\
		& & &\text{ for all } 0 \leq m < l: (\Pi,j_{m},j_{m+1}) \in R(\alpha) \\
		&(\Pi,i,k) \in R(\varphi ?) &\text{iff }   & i = k \text{ and } \Pi[i,\infty] \models_{\mathcal{T}} \varphi
	\end{align*}
\end{definition}

%% file: sections/main/modelchecking.tex
\section{Model Checking HyperPDL-$\Delta$}\label{Section:ModelChecking}
In order to tackle the model checking problem for HyperPDL-$\Delta$, that is to check whether $\mathcal{T} \models \varphi$ holds for arbitrarily given KTS $\mathcal{T}$ and closed HyperPDL-$\Delta$ formulas $\varphi$, we develop a new algorithm inspired by the HyperCTL$^*$ model checking algorithm from \cite{Finkbeiner2015}.
A crucial idea is to represent a path assignment $\Pi$ with $\dom(\Pi) = \{\epsilon,\pi_1,...,\pi_n\}$ by an $\omega$-word over $(S^n \times \Sigma^n)$.
Formally, we define a translation function $\nu: \textit{PA} \to \bigcup_{n \in \mathbb{N}_0}(S^n \times \Sigma^n)^{\omega}$ such that a path assignment $\Pi$ with $\Pi(\pi_i) = s_i^0 \sigma_i^0 s_i^1 \sigma_i^1 ...$ is mapped to $\nu(\Pi) = ((s_0^0,...,s_n^0),(\sigma_0^0,...,\sigma_n^0))((s_0^1,...,s_n^1),(\sigma_0^1,...,\sigma_n^1))... \in (S^n \times \Sigma^n)^{\omega}$ for $n = |\dom(\Pi)|$.
Note that $\epsilon$ need not be encoded separately in $\nu(\Pi)$ since $\Pi(\pi_n) = \Pi(\epsilon)$ always holds.
Similar to the notation for $\tau$ used before, we use the notation $\mathpzc{s}$ to refer to tuples $(s_1,...,s_n)$ and write $\mathpzc{s}|_i$ to refer to $s_i$.
Then, given a formula $\varphi$ and a KTS $\mathcal{T}$, we construct an ABA $\mathcal{A}_{\varphi}$ recognising $\nu(\Pi)$ iff $\Pi \models_{\mathcal{T}} \varphi$ holds.

First, we transform all formulas into a variation of negation normal form, where only existential quantifiers are allowed and negation can only occur in front of existential quantifiers, atomic propositions or $\Delta$s.
This form differs from conventional NNF in that negated existential instead of universal quantifiers are used, because they can be handled more efficiently in our setup.
The transformation is done by driving all negations in the formula inwards using De Morgan's laws and the duality $\lnot \langle \alpha \rangle \varphi \equiv [ \alpha ] \lnot \varphi$ while also replacing universal quantifiers $\forall$ with $\lnot \exists \lnot$ and cancelling double negations successively.
For example, the formula $\forall \pi . [\alpha] \lnot a_{\pi}$ is transformed into $\lnot \exists \pi . \lnot [\alpha] \lnot a_{\pi}$, $\lnot \exists \pi . \langle \alpha \rangle \lnot \lnot a_{\pi}$ and then $\lnot \exists \pi . \langle \alpha \rangle a_{\pi}$ successively.

Then, as another preprocessing step, all programs $\alpha$ appearing in the formula are inductively translated to an intermediate automaton representation $M_{\alpha} = (Q, q_0, \Sigma^n, \rho, q_f, \Psi)$.
The automaton $M_{\alpha}$ can be seen as a nondeterministic finite automaton (NFA) with access to oracles for the tests in $\alpha$ which recognises all prefixes up to index $j$ of the $\omega$-word $\nu(\Pi[i,\infty])$ such that $(\Pi,i,j) \in R(\alpha)$.
This is formalised in \autoref{constructionlemma}.
The only differences in syntax when compared to a conventional NFA are (i) there is exactly  one final state $q_f$ instead of a set $F$, (ii) $\varepsilon$-edges are not eliminated and (iii) we have a state marking function $\Psi$ mapping every state $q \in Q$ to a singleton or empty set of formulas $\Psi(q)$.
State markings $\Psi(q)$ are introduced to tackle tests $\psi ?$ and are later replaced by transitions to automata $\mathcal{A}_{\psi}$, which we define in the construction for formulas.
These state markings make the standard elimination of $\varepsilon$-edges impossible, which is why we delay the elimination until the markings are eliminated as well.
This approach is similar to the one used in \cite{Faymonville2017} to transform LDL formulas into automata.
However, we apply the construction in a hyperlogic context and make it more succinct by  offering an alternative approach to constructing the transition function, thus avoiding an exponential blowup.

\input{math/constructions/alpha}

Let $\xRightarrow{\varepsilon}_X \subseteq Q \times Q$ for a set of formulas $X$ be the smallest relation such that (i) $q \xRightarrow{\varepsilon}_{\Psi(q)} q$ and (ii) $q' \xRightarrow{\varepsilon}_{X} q''$ and $q' \in \rho(q,\varepsilon)$ imply $q \xRightarrow{\varepsilon}_{X \cup \Psi(q)} q''$.
Then, for $\tau \in \Sigma^n$, let $\xRightarrow{\tau}_X$ be the smallest relation such that $q \xRightarrow{\tau}_{X} q''$ iff there is a $q' \in Q$ such that $q \xRightarrow{\varepsilon}_{X} q'$ and $q'' \in \rho(q',\tau)$.
These relations capture the encountered markings along $\varepsilon$-paths in $M_{\alpha}$ in the following way: $q \xRightarrow{\varepsilon}_{X} q'$ holds if there is an $\varepsilon$-path from $q$ to $q'$ in $M_{\alpha}$ that encounters exactly the state markings in the set $X$.
$q \xRightarrow{\tau}_{X} q'$ is used to describe the same behaviour, but requires an additional $\tau$-step at the end.

We obtain the following Lemma:

\input{math/theorems/constructionlemma}

As mentioned, we transform $\varphi$ into an ABA $\mathcal{A}_{\varphi}$ recognising $\nu(\Pi)$ iff $\Pi \models_{\mathcal{T}} \varphi$ holds.
Formally, a language $\mathcal{L} \subseteq (S^n \times {\Sigma}^n)^{\omega}$ is called $\mathcal{T}$-equivalent to a formula $\varphi$, if for each $\Pi$ the statements $\nu(\Pi) \in {\mathcal{L}}(\mathcal{A}_{\varphi})$ and $\Pi \models_{\mathcal{T}} \varphi$ are equivalent; we say that an ABA $\mathcal{A}$ is $\mathcal{T}$-equivalent to $\varphi$ iff its language $\mathcal{L}(\mathcal{A})$ is $\mathcal{T}$-equivalent to $\varphi$.
As a closed formula $\varphi$ is a boolean combination of quantified subformulas, the model checking problem can be solved by performing separate nonemptiness checks on $\mathcal{A}_{\psi}$ for all maximal quantified subformulas $\psi$ and combining the results in accordance with the global structure of $\varphi$.
For example, if $\varphi$ is given as $\psi_1 \land \lnot \psi_2$ for quantified formulas $\psi_1$ and $\psi_2$, one performs emptiness tests on $\mathcal{A}_{\psi_1}$ as well as $\mathcal{A}_{\psi_2}$ and accepts iff the first test is positive and the second test is negative.

\input{math/constructions/modelchecking}

For the construction, we obtain the following theorem via induction:

\input{math/theorems/construction}

For the complexity analysis, we introduce a notion of criticality.

\input{math/defs/criticality}

This definition is designed carefully in order to ensure that the alternation depth of HyperCTL$^*$ formulas coincides with the criticality of their direct translation to HyperPDL-$\Delta$.
Intuitively, an \textit{uncritical} quantifier is one where the next dealternation construction $\textit{MH}(\mathcal{A})$ does not cause another exponential blowup on this part of the automaton.
Accordingly, a critical quantifier is one where this exponential blowup for the next dealternation construction cannot be avoided in general.
Thus, the criticality of a formula accounts for the number of times an exponential blowup may happen during the construction.


Since exponential blowups occur in a nested manner, we can only bound the size of the resulting automaton by an exponential tower.
As argued in the last paragraph, its height is determined by criticality rather than quantifier depth.
Formally, we define a function $g$ as $g_{p,c}(0,n) = p(n)$ and $g_{p,c}(k+1,n) = c^{g(k,n)}$ for a constant $c > 1$ and a polynomial $p$.
We use $\mathcal{O}(g(k,n))$ as an abbreviation for $\mathcal{O}(g_{p,c}(k,n))$ for some $c > 1$ and polynomial $p$.

Two remarks are in order about the next Lemma.
Firstly, the statement refers to the number of states of the construction, disregarding the number of transitions.
However, this is harmless for our complexity analysis: while the alphabet size increases exponentially with the nesting depth of quantifiers, alphabets need not be represented explicitly and the number of transitions of the automata can be kept polynomial by delaying the substitution of the wildcard symbol $\cdot$ by concrete programs until the intersection with the system $\mathcal{T}$.
We refrain from explicating this in our construction in order to increase readability.
Secondly, the construction's size can also increase exponentially in the nesting level of negated $\Delta \alpha$ constructions.
Since we expect that negated $\Delta \alpha$ constructions are rarely nested in formulas, we assume for the remainder of this paper a bound on this nesting level in order to simplify our complexity statement.
Indeed, in \autoref{section:expressivity}, we will show that a bound of $1$ suffices to express $\omega$-regular properties, thus making this a reasonable constraint.
The dependency from the nesting depth is reflected in the proof of \autoref{complexitylemma} in \autoref{appendix:complexitylemma}.

\input{math/theorems/complexitylemma}

\input{math/theorems/complexity}

Since we can easily translate HyperCTL$^*$ formulas to HyperPDL-$\Delta$ while preserving the alternation depth as criticality, we can use known hardness results for HyperCTL$^*$ model checking \cite{Rabe2016} to obtain the following Theorem:

\input{math/theorems/hardness}

%% file: math/constructions/alpha.tex
\textbf{Construction of \boldmath$M_{\alpha}$\unboldmath:} We now describe the construction of $M_{\alpha}$.
When an expression $\alpha$ has one or more subexpressions $\alpha_i$, we assume that automata $M_{\alpha_i} = (Q_i, q_{0,i}, {\Sigma}^n, \rho_i, q_{f,i}, \Psi_i)$ are already constructed.
Intuitively, most cases are analogous to the translation from regular expressions to NFA.
For the case of a test $\psi ?$, a state marked with $\psi$ is introduced in between starting and final state.
Detailed constructions are shown in \autoref{construction:alpha}.

\begin{figure}[t]
	\begin{tabular}{l l}
		\boldmath$\tau$\unboldmath & $M_{\alpha} = (\{q_0, q_1\}, q_0, {\Sigma}^n, \rho, q_1, \Psi)$, $\Psi(q_i) = \emptyset$ \\
		& $\rho(q,(\sigma_1,\dots,\sigma_n)) = \begin{cases}
		\{q_1\} &\text{if } \forall i . \tau|_i = \cdot \lor \tau|_i = \sigma_i \text{ and } q = q_0 \\
		\emptyset &\text{else}
		\end{cases}$\\
		& $\rho(q,\varepsilon) = \emptyset$ \\
		\boldmath$\varepsilon$\unboldmath & $M_{\alpha} = (\{q_0,q_1\},q_0,{\Sigma}^n, \rho,q_1,\Psi)$, $\Psi(q_i) = \emptyset$ \\
		& $\rho(q_i,\tau) = \emptyset$ \\
		& $\rho(q_i,\varepsilon) = \begin{cases}
		\{q_{i+1}\} &\text{if } i = 0 \\
		\emptyset &\text{else}
		\end{cases}$\\
		\boldmath$\alpha_1 + \alpha_2$\unboldmath & $M_{\alpha} = (Q_1 \dot\cup Q_2 \dot\cup \{q_0,q_f\}, q_0, {\Sigma}^n, \rho, q_f, \Psi)$,\\
		& $\Psi(q_0) = \Psi(q_f) = \emptyset$, $\Psi(q) = \Psi_i(q)$ for $q \in Q_i$ \\
		& $\rho(q,\tau) = \begin{cases}
		\rho_i(q,\tau) &\text{if } q \in Q_i \\
		\emptyset &\text{else}
		\end{cases}$\\
		& $\rho(q,\varepsilon) = \begin{cases}
		\{q_{0,1},q_{0,2}\} &\text{if } q = q_0 \\
		\rho_i (q,\varepsilon) &\text{if } q \in Q_i \setminus \{q_{f,i}\} \\
		\rho_i (q,\varepsilon) \cup \{q_f\} &\text{if } q = q_{f,i}
		\end{cases}$ \\
		\boldmath$\alpha_1 \cdot \alpha_2$\unboldmath & $M_{\alpha} = (Q_1 \dot\cup Q_2, q_{0,1}, {\Sigma}^n, \rho, q_{f,2}, \Psi)$, $\Psi(q) = \Psi_i(q)$ for $q \in Q_i$ \\
		& $\rho(q,\tau) = \rho_i(q,\tau)$ for $q \in Q_i$ \\
		& $\rho(q,\varepsilon) = \begin{cases}
		\rho_i(q,\varepsilon) &\text{if } q \in Q_i, q \neq \{q_{f,1}\} \\
		\rho_1(q,\varepsilon) \cup \{q_{0,2}\} &\text{if } q = q_{f,1}
		\end{cases}$\\
		\boldmath$(\alpha_1)^*$\unboldmath & $M_{\alpha} = (Q_1 \dot\cup \{q_0,q_f\}, q_0, {\Sigma}^n, \rho, q_f, \Psi)$ ,$\Psi(q_0) = \Psi(q_f) = \emptyset$, $\Psi(q) = \Psi_1(q)$, for $q \in Q_1$ \\
		& $\rho(q,\tau) = \begin{cases}
		\rho_1(q,\tau) & \text{if } q \in Q_1 \\
		\emptyset & \text{else}
		\end{cases}$\\
		& $\rho(q,\varepsilon) = \begin{cases}
		\rho_1(q,\varepsilon) &\text{if } q \not\in \{q_0,q_f,q_{f,1}\} \\
		\{q_{0,1},q_f\} &\text{if } q = q_0 \\
		\{q_0\}& \text{if } q = q_f\\
		\rho_1(q,\varepsilon) \cup \{q_f\} &\text{if } q = q_{f,1}
		\end{cases}$\\
		\boldmath$\psi ?$\unboldmath & $M_{\alpha} = (\{q_0, q_1, q_2\}, q_0, {\Sigma}^n, \rho, q_2, \Psi)$,$\Psi(q_0) =\Psi(q_2) = \emptyset$, $\Psi(q_1) = \{\psi\}$ \\
		& $\rho(q,\tau) = \emptyset$ \\
		& $\rho(q_i,\varepsilon) = \begin{cases}
		\{q_{i+1}\} &\text{if } i = 0,1\\
		\emptyset &\text{else}
		\end{cases}$
	\end{tabular}
	\caption{Construction of $M_{\alpha}$}
	\label{construction:alpha}
\end{figure}

%% file: math/theorems/constructionlemma.tex
\begin{lemma}\label{constructionlemma}
	Let $\Pi$ be a path assignment with $\nu(\Pi) = (\mathpzc{s}_0\tau_1)(\mathpzc{s}_2\tau_3) ...$ and $\alpha$ a program, then $(\Pi,i,k) \in R(\alpha)$ iff there is a state sequence $q_0 q_1 ... q_m$ with $m = \frac{k-i}{2}$ in $M_{\alpha}$ and sets  of formulas $X_0,...,X_m$ such that \\
	(i) $q_0$ is the initial state of $M_{\alpha}$, \\
	(ii) $q_m \xRightarrow{\varepsilon}_{X_m} q_f$ for the final state $q_f$ of $M_{\alpha}$, \\
	(iii) $q_{l} \xRightarrow{\tau_{i+2l+1}}_{X_l} q_{l+1}$ for all $l<m$ and \\
	(iv) $\psi \in X_l$ implies $\Pi[i+2l,\infty] \models_{\mathcal{T}} \psi$ for all $l \leq m$.
\end{lemma}

%% file: math/constructions/modelchecking.tex
\textbf{Construction of \boldmath$\mathcal{A}_{\varphi}$\unboldmath:} We construct the ABA $\mathcal{A}_{\varphi}$ inductively.
The alphabet of $\mathcal{A}_{\varphi}$ is given as $\Sigma_{\varphi} = S^n \times {\Sigma}^n$ where $n$ is the number of path quantifiers the formula $\varphi$ is in scope of.
When constructing an automaton for $\varphi$ with subformulas $\varphi_i$, we assume that the automata $\mathcal{A}_{\varphi_{i}} = (Q_{\varphi_{i}},q_{0,\varphi_{i}},\Sigma_{\varphi_i},\rho_{\varphi_{i}},F_{\varphi_{i}})$ are already constructed.
Similarly, when a formula contains an expression $\alpha$, we assume that not only $M_{\alpha} = (Q_{\alpha},q_{0,\alpha},{\Sigma}^n, \rho_{\alpha}, q_{f,\alpha}, \Psi)$ but also $\mathcal{A}_{\psi_i}$ in the case of $\langle \alpha \rangle \varphi$ and $\Delta \alpha$ or $\mathcal{A}_{\bar{\psi_i}}$ in case of $[\alpha]\varphi$ for each construct $\psi_i ?$ in $\alpha$ are already constructed.
Here, $\bar{\psi_i}$ is the negation normal form of $\lnot \psi_i$.
Recall that states $\true$ and $\false$ are always part of an ABA in our definition, so we will not mention them explicitly.
Furthermore, let $\mathcal{T} = (S,s_0,\{\delta_{\sigma}\mid\sigma \in \Sigma\},L)$ be the KTS to be checked.
\begin{figure}[t]
	\begin{tabular}{l l}
		\boldmath$a_{\pi_k}$\unboldmath & $\mathcal{A}_{\varphi} = (\{q_0\}, q_0, \Sigma_{\varphi}, \rho , \emptyset)$ \\
		& $\rho(q_0,(\mathpzc{s},\tau)) = \begin{cases}
		\true &\text{if } a \in L(\mathpzc{s}|_k)\\
		\false &\text{else}
		\end{cases}$ \\
		\boldmath$\lnot a_{\pi_k}$\unboldmath & $\mathcal{A}_{\varphi} = (\{q_0\}, q_0, \Sigma_{\varphi}, \rho , \emptyset)$ \\
		& $\rho(q_0,(\mathpzc{s},\tau)) = \begin{cases}
		\false &\text{if } a \in L(\mathpzc{s}|_k)\\
		\true &\text{else}
		\end{cases}$ \\
		\boldmath$\varphi_1 \land \varphi_2$\unboldmath & $\mathcal{A_{\varphi}} = (Q_1 \dot\cup Q_2 \dot\cup \{q_0\},q_{0},\Sigma_{\varphi},\rho,F_1 \dot\cup F_2)$ \\
		& $\rho(q,(\mathpzc{s},\tau)) = \begin{cases}
		\rho_1(q_{0,1},(\mathpzc{s},\tau)) \land \rho_2(q_{0,2},(\mathpzc{s},\tau)) &\text{if } q = q_0 \\
		\rho_i(q,(\mathpzc{s},\tau)) &\text{if } q \in Q_i, i \in \{1,2\}
		\end{cases}$ \\
		\boldmath$\varphi_{1} \lor \varphi_{2}$\unboldmath & $\mathcal{A_{\varphi}} = (Q_1 \dot\cup Q_2 \dot\cup \{q_0\},q_{0},\Sigma_{\varphi},\rho,F_1 \dot\cup F_2)$ \\
		& $\rho(q,(\mathpzc{s},\tau)) = \begin{cases}
		\rho_1(q_{0,1},(\mathpzc{s},\tau)) \lor \rho_2(q_{0,2},(\mathpzc{s},\tau)) &\text{if } q = q_0 \\
		\rho_i(q,(\mathpzc{s},\tau)) &\text{if } q \in Q_i, i \in \{1,2\}
		\end{cases}$
	\end{tabular}
	\caption{Construction of $\mathcal{A}_{\varphi}$ in basic cases}
	\label{construction:phi:basic}
\end{figure}

The cases of the constructions shown in \autoref{construction:phi:basic} are straightforward:
for $a_{\pi_k}$ (resp. $\lnot a_{\pi_k}$), only those words are accepted where the state first read for path $\pi_k$ is labelled with $a$ (resp. not labelled with $a$).
For the connectives $\land$ and $\lor$, one checks whether both automata $\mathcal{A}_{\varphi_{1}}$ and $\mathcal{A}_{\varphi_{2}}$ accept (resp. at least one automaton accepts) the word.

\begin{figure}
	\begin{tabular}{l l}
		\boldmath$\langle \alpha \rangle \varphi_1$\unboldmath & $\mathcal{A}_{\varphi} = (Q_1 \dot\cup Q_{\alpha} \dot\cup \bigcup_i Q_{\psi_i},q_{0,\alpha},\Sigma_{\varphi},\rho,F_1 \dot\cup \bigcup_i F_{\psi_i})$ \\
		& $\rho(q,(\mathpzc{s},\tau)) = \begin{cases}
		\rho_1(q,(\mathpzc{s},\tau)) &\text{if } q \in Q_1 \\
		\rho_{\psi_i}(q,(\mathpzc{s},\tau)) &\text{if } q \in Q_{\psi_i} \\
		\bigvee \{q' \land \bigwedge_{\psi_i \in X} \rho_{\psi_i}(q_{0,\psi_i},(\mathpzc{s},\tau)) \mid q \xRightarrow{\tau}_{X} q'\}\ \cup \\
		\quad \{\rho(q_{0,1}, (\mathpzc{s}, \tau)) \land \\
		\qquad\qquad \bigwedge_{\psi_i \in X} \rho_{\psi_i}(q_{0,\psi_i},(\mathpzc{s},\tau)) \mid q \xRightarrow{\varepsilon}_X q_{f,\alpha} \} &\text{if } q \in Q_{\alpha} 
		\end{cases}$ \\
		\boldmath$[\alpha] \varphi_{1}$\unboldmath & $\mathcal{A}_{\varphi} = (Q_1 \dot\cup Q_{\alpha} \dot\cup \bigcup_i Q_{\bar{\psi_i}},q_{0,\alpha},\Sigma_{\varphi},\rho,F_1 \dot\cup Q_{\alpha} \dot\cup \bigcup_i F_{\bar{\psi_i}})$ \\
		& $\rho(q,(\mathpzc{s},\tau)) = \begin{cases}
		\rho_1(q,(\mathpzc{s},\tau)) &\text{if } q \in Q_1 \\
		\rho_{\bar{\psi_i}}(q,(\mathpzc{s},\tau)) &\text{if } q \in Q_{\bar{\psi_i}} \\
		\bigwedge \{q' \lor \bigvee_{\psi_i \in X} \rho_{\bar{\psi_i}}(q_{0,\bar{\psi_i}},(\mathpzc{s},\tau)) \mid q \xRightarrow{\tau}_{X} q'\}\ \cup \\ 
		\quad \{\rho(q_{0,1}, (\mathpzc{s}, \tau)) \lor \\
		\qquad\qquad \bigvee_{\psi_i \in X} \rho_{\bar{\psi_i}}(q_{0,\bar{\psi_i}},(\mathpzc{s},\tau)) \mid q \xRightarrow{\varepsilon}_X q_{f,\alpha} \} &\text{if } q \in Q_{\alpha} 
		\end{cases}$ \\
		\boldmath$\Delta \alpha$\unboldmath & $\mathcal{A}_{\varphi} = (Q_{\alpha} \dot\cup \{q_0\} \dot\cup \bigcup_i Q_{\psi_i},q_0,\Sigma_{\varphi},\rho,\{q_0\} \dot\cup \bigcup_i F_{\psi_i})$ \\
		& $\rho(q,(\mathpzc{s},\tau)) = \begin{cases}
		\rho_{\psi_i}(q,(\mathpzc{s},\tau)) &\text{if } q \in Q_{\psi_i} \\
		\bigvee \{q' \land \bigwedge_{\psi_i \in X} \rho_{\psi_i}(q_{0,\psi_i},(\mathpzc{s},\tau)) \mid q_{0,\alpha} \xRightarrow{\tau}_{X} q'\} \cup \\
		\quad \{ \bigwedge_{\psi_i \in X} \rho_{\psi_i}(q_{0,\psi_i},(\mathpzc{s},\tau)) \mid q_{0,\alpha} \xRightarrow{\varepsilon}_X q_{f,\alpha} \} &\text{if } q = q_0 \\
		\bigvee \{q' \land \bigwedge_{\psi_i \in X} \rho_{\psi_i}(q_{0,\psi_i},(\mathpzc{s},\tau)) \mid q \xRightarrow{\tau}_{X} q'\}\ \cup \\ \quad \{q_0 \land \bigwedge_{\psi_i \in X} \rho_{\psi_i}(q_{0,\psi_i},(\mathpzc{s},\tau)) \land \\
		\qquad\qquad \bigwedge_{\psi_i \in Y} q_{0,\psi_i} \mid q \xRightarrow{\tau}_X q' \xRightarrow{\varepsilon}_Y q_{f,\alpha} \} &\text{if } q \in Q_{\alpha}
		\end{cases}$\\
		\boldmath$\lnot \Delta \alpha$\unboldmath & $\mathcal{A}_{\varphi} = \overline{\mathcal{A}_{\Delta \alpha}}$
	\end{tabular}
	\caption{Construction of $\mathcal{A}_{\varphi}$ for $\alpha$ formulas}
	\label{construction:phi:alpha}
\end{figure}

\begin{figure}
	\centering
	\begin{tikzpicture}
	
	\node[draw,dashed,ellipse,minimum height=30pt] at (-3,-1) (palpha1){$M_{\alpha}$, $\bigvee$};
	\node[draw,circle,fill=white!10] at (-3,-1.5) (palphas){};
	\node[draw,double,circle,fill=white!10] at (-2.3,-0.6) (palphaf){};
	\node[draw,circle,fill=white!10] at (-3.7,-0.6) (palphapsi){};
	
	\node[draw,ellipse,minimum height=30pt] at (-4,1) (ppsi){$\mathcal{A}_{\psi}$};
	\node[draw,circle,fill=white!10] at (-4,0.5) (ppsis){};
	
	\node[draw,ellipse,minimum height=30pt] at (-2,1) (pvarphi){$\mathcal{A}_{\varphi}$};
	\node[draw,circle,fill=white!10] at (-2,0.5) (pvarphis){};
	
	\node[] at (-4.1,-0.7) (pmark){$\psi$};
	\node[] at (-4.1,-0.3) (pedge1){$\land$};
	\node[] at (-1.9,-0.3) (pedge2){$\lor$};
	
	\path[->] (-3,-2) edge (palphas);
	\path[->] (palphapsi) edge[bend left=20pt] (ppsis);
	\path[->] (palphaf) edge[bend right=20pt] (pvarphis);
	
	
	\node[draw,dashed,double,ellipse,minimum height=30pt] at (3,-1) (qalpha1){$M_{\alpha}$, $\bigwedge$};
	\node[draw,circle,fill=white!10] at (3,-1.5) (qalphas){};
	\node[draw,double,circle,fill=white!10] at (3.7,-0.6) (qalphaf){};
	\node[draw,circle,fill=white!10] at (2.3,-0.6) (qalphapsi){};
	
	\node[draw,ellipse,minimum height=30pt] at (2,1) (qpsi){$\mathcal{A}_{\lnot\psi}$};
	\node[draw,circle,fill=white!10] at (2,0.5) (qpsis){};
	
	\node[draw,ellipse,minimum height=30pt] at (4,1) (qvarphi){$\mathcal{A}_{\varphi}$};
	\node[draw,circle,fill=white!10] at (4,0.5) (qvarphis){};
	
	\node[] at (1.9,-0.7) (qmark){$\psi$};
	\node[] at (1.9,-0.3) (qedge1){$\lor$};
	\node[] at (4.1,-0.3) (qedge2){$\land$};
	
	\path[->] (3,-2) edge (qalphas);
	\path[->] (qalphapsi) edge[bend left=20pt] (qpsis);
	\path[->] (qalphaf) edge[bend right=20pt] (qvarphis);
	\end{tikzpicture}
	\caption{Illustration of the constructions for $\langle \alpha \rangle \varphi$ and $[\alpha] \varphi$.
	In automata shown with a single dashed line all states are non-final.
	In automata shown with a double dashed line all states are final.}
	\label{illustration:phi:alpha}
\end{figure}
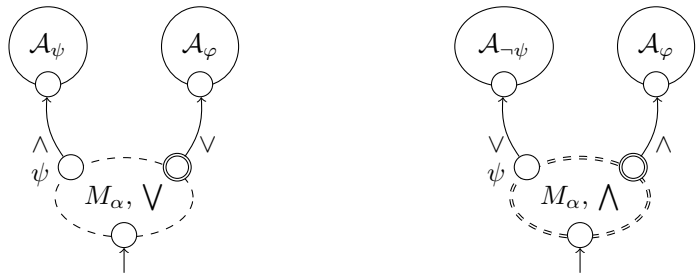

We now discuss how to handle the PDL-like modalities $\langle \alpha \rangle \varphi$, $[\alpha]\varphi$ and $\Delta \alpha$.
In the correctness statement for the automata $M_{\alpha}$ constructed for this purpose, we rely on oracle requests for tests $\psi?$ (\autoref{constructionlemma} (iv)). 
As mentioned, we eliminate these oracle requests by transitioning into the automaton $\mathcal{A}_{\psi}$ whenever we reach a state marked with $\psi ?$.

In the construction for $\langle \alpha \rangle \varphi_{1}$, we want to recognise a single path where after a prefix satisfying $\alpha$, $\varphi_{1}$ holds.
This is achieved by enabling a move into $\mathcal{A}_{\varphi_{1}}$ whenever the final state $q_{f,\alpha}$ of $M_{\alpha}$ is reached.
Since none of the states of $M_{\alpha}$ are declared final in $\mathcal{A}_{\varphi}$, an accepting run in $\mathcal{A}_{\varphi}$ cannot stay in $M_{\alpha}$ forever and thus eventually has to move into $\mathcal{A}_{\varphi_{1}}$ in this way.
Moves into $\mathcal{A}_{\psi}$ for tests $\psi ?$ are made conjunctively since all tests on an accepting run have to be successful.
For the dual case $[\alpha] \varphi_{1}$, we want $\varphi_{1}$ to hold after all prefixes satisfying $\alpha$.
Thus, dual to the previous construction, whenever reaching the state $q_{f,\alpha}$, we are not only given the possibility to, but have to move into $\mathcal{A}_{\varphi_{1}}$ as well.
Since  all transitions in this construction are combined with $\land$ to ensure that all prefixes satisfying $\alpha$ are considered, we have to take care of paths in $\mathcal{A}_{\varphi}$ that never leave $M_{\alpha}$ by declaring all states of $M_{\alpha}$ accepting.
On the other hand, paths not satisfying $\alpha$ due to a violation of a test $\psi ?$ are ruled out by a disjunctive test for the negation of $\psi$.
An illustration of these two cases can be found in \autoref{illustration:phi:alpha}.
To handle formulas of the form $\Delta \alpha$, we transform $M_{\alpha}$ into a Büchi automaton with a distinguished new initial state $q_0$ which ensures that in between two visits of $q_0$, $\alpha$ is satisfied. 
The state $q_0$ acts like the initial state $q_{0,\alpha}$ for outgoing, and like the final state $q_{f,\alpha}$ for incoming transitions and is the only accepting state (apart from those in the $\mathcal{A}_{\psi_i}$ automata).
Moves into test-automtata are handled just as in the $\langle \alpha \rangle \varphi_1$ case.
This ensures that an accepted input word consists of repeated segments matched by $\alpha$.
Negated $\Delta \alpha$ constructions can be handled by complementation using \autoref{complementation}.

Note that when we translate state markings in $M_{\alpha}$ into transitions in $\mathcal{A}_{\varphi}$, $q \xRightarrow{\varepsilon}_X q'$ may hold for exponentially many sets $X$, resulting in disjunctions of exponential size (in $|\alpha|$) in the transition functions for $\langle \cdot \rangle$ and $\Delta$, and conjunctions of exponential size for $[\cdot]$.
This can, however, be avoided by constructing formulas equivalent to these disjunctions and conjunctions during the inductive construction of $M_{\alpha}$, the size of which is not larger than $3 \cdot |\alpha| + 2$.
We discuss the case of disjunctions; the other case can be handled in a dual way.
Two observations are exploited.
First of all, the formula need not be written in a disjunctive form but can mix disjunctions and conjunctions freely. 
Thus, one source of exponential increase can be avoided by constructing the formulas for nested sums and concatenations inductively using that their contribution is directly captured by disjunction and conjunction, respectively.
The second observation is that the conjunction resulting from a subpath of a path subsumes the conjunction for that path.
This can be exploited to show that only paths using backwards edges (i.e., edges from $q_f$ to $q_0$ in *-constructions) at most once and using only particular backward edges have to be considered for treating iteration.
We refer the interested reader to \autoref{appendix:alphaalt} for a more detailed look at this alternative construction.


\begin{figure}
	\begin{tabular}{l l}
		\boldmath$\exists \pi. \varphi_1$\unboldmath & $\mathcal{A}_{\varphi_{1}}$ dealternised: $\textit{MH}(\mathcal{A}_{\varphi_1}) = (Q_1,q_{0,1},\Sigma_{\varphi_1}, \rho_1, F_1)$ \\
		& $\mathcal{A}_{\varphi} = (Q_1 \times S \times \Sigma \dot\cup \{q_0\},q_0,\Sigma_{\varphi},\rho,F_1 \times S \times \Sigma)$ \\
		& $\rho(q_0,(\mathpzc{s},\tau)) = \{(q',s',\sigma ')\mid q' \in \rho_1(q_{0,1},\mathpzc{s} + \mathpzc{s}|_n , \tau + \sigma),s' \in \delta_{\sigma}(\mathpzc{s}|_n), \sigma, \sigma' \in \Sigma\}$ \\
		& $\rho((q,s,\sigma),(\mathpzc{s},\tau)) = \{(q',s',\sigma ')\mid q' \in \rho_1(q,\mathpzc{s} + s , \tau + \sigma), s' \in \delta_{\sigma}(s), \sigma ' \in \Sigma\}$ \\
		\boldmath$\lnot \exists \pi. \varphi_1$\unboldmath & $\mathcal{A}_{\varphi} = \overline{\mathcal{A}_{\exists \pi. \varphi_{1}}}$
	\end{tabular}
	\caption{Construction of $\mathcal{A}_{\varphi}$ for quantifier formulas}
	\label{construction:phi:quantifier}
\end{figure}

In the constructions for path quantifiers (\autoref{construction:phi:quantifier}), we eliminate one component of the alphabet $\Sigma_{\varphi_1} = S^{n+1} \times \Sigma^{n+1}$ and switch to $\Sigma_{\varphi} = S^n \times \Sigma^n$.
The eliminated component is now simulated by the state space of $\mathcal{A}_{\varphi}$, which ensures that said component is indeed a path in $\mathcal{T}$ by using the additional components from $S$ and $\Sigma$.
The rule for the initial state guarantees that this path from $\mathcal{T}$ starts in the state it is branching from.
Negated existential quantifiers that by our definition can occur in formulas in negation normal form are handled straightforwardly by complementation.

%% file: math/theorems/construction.tex
\begin{theorem}\label{construction}
	The automaton $\mathcal{A}_{\varphi}$ is $\mathcal{T}$-equivalent to $\varphi$.
\end{theorem}

%% file: math/defs/criticality.tex
\begin{definition}[Criticality]
	The \emph{criticality} of a HyperPDL-$\Delta$ formula $\varphi$ in negation normal form equals the highest number of critical quantifiers along any path in the formula's syntax tree.
	A quantifier is called \emph{critical} iff it is a non-outermost quantifier, fulfills at least one of the following three conditions: \\
	(i) it is a negated quantifier, \\
	(ii) it is an outermost quantifier in a test $\varphi ?$ in some program $\alpha$, \\
	(iii) it is an outermost quantifier in the subformula $\varphi$ of $[\alpha]\varphi$,\\
        and is not a negated outermost quantifier in a test $\varphi ?$ occuring in a modality $[\alpha]$ where the automaton $M_{\alpha}$ is deterministic.\\
	We call a HyperPDL-$\Delta$ formula with criticality $0$ \emph{uncritical}.
\end{definition}

%% file: math/theorems/complexitylemma.tex
\begin{lemma}\label{complexitylemma}
	The automaton $\mathcal{A}_{\varphi}$ has size $\mathcal{O}(g(k+1,|\varphi| + log(|\mathcal{T}|)))$ for formulas $\varphi$ with criticality $k$.
\end{lemma}

\begin{proof}(Sketch)
	By induction on the criticality $k$.
	In the base case, the dealternation constructions from \autoref{dealternation} overall cause a single exponential blowup at most since in all constructions not increasing the criticality, one has to keep track of only a single state of each already dealternised subautomaton in further dealternations.
	Since the only dealternation is done before the system $\mathcal{T}$ is folded around the automaton, the automaton's size is only exponential in the size of $\varphi$, but not in the size of $\mathcal{T}$.
	
	In the inductive step, the construction for the outermost critical quantifier increases the size of the automaton exponentially.
	For all further constructions, it can be argued that the automaton's size is asymptotically not further increased, just like in the base case.
\end{proof}

%% file: math/theorems/complexity.tex
\begin{theorem}\label{complexity}
	The problem to decide whether $\mathcal{T} \models \varphi$ holds for KTS $\mathcal{T}$ and HyperPDL-$\Delta$ formulas $\varphi$ with criticality $k$ is in NSPACE$(g(k,|\varphi| + log(|\mathcal{T}|))$.
\end{theorem}

\begin{proof}
	The formula arising from the transformation of $\varphi$ to our variant of negation normal form is a boolean combination of subformulas $\psi$ with an outermost existential quantifier. 
	Due to dealternation for the existential quantifier, the automata $\mathcal{A}_{\psi}$ are non-deterministic Büchi automata.
	We perform nonemptiness checks on these automata separately.
	It is well-known that the nonemptiness check for Büchi automata is possible in NLOGSPACE in the size of the automaton \cite{Lange2016}.
	We then combine the results in accordance with the structure of $\varphi$.
	This does not add to the complexity.
	Thus, by \autoref{complexitylemma}, we obtain an NSPACE$(g(k,|\varphi| + log(|\mathcal{T}|))$ model checking algorithm for criticality $k$ HyperPDL-$\Delta$ formulas.
\end{proof}

%% file: math/theorems/hardness.tex
\begin{theorem}\label{hardness}
	Given a KTS $\mathcal{T}$ and a HyperPDL-$\Delta$ formula $\varphi$ with criticality $k$, the model checking Problem for HyperPDL-$\Delta$ is hard for NSPACE$(g(k, |\varphi|))$ and NSPACE$(g(k-1, |\mathcal{T}|))$.\footnote{Note that we have not defined $g(-1,n)$ in this paper.
	For $k=0$ and a fixed size formula $\varphi$, we use the definition from \cite{Finkbeiner2015}, where NSPACE$(g(-1, n))$ was defined as NLOGSPACE.
	Since one can see that our algorithm has NLOGSPACE complexity in this instance, the lower and upper bounds match in all cases.
	For $k = 0$ and a fixed size structure $\mathcal{T}$ or $k > 1$, we can use Savitch's Theorem to see that the problems are actually complete for the deterministic space classes.}
\end{theorem}

%% file: sections/main/satisfiability.tex
\section{Satisfiability}\label{section:satisfiability}
While model checking of temporal logics is an essential technique for verification, it requires meaningful specifications in order to be useful.
For example, if a formula is fulfilled by every or no structure, it is useless for specification purposes.
To evaluate whether this is the case, satisfiability testing can be employed as a sanity check \cite{Rozier2007}. 
Since satisfiability checking is already undecidable for HyperLTL \cite{Finkbeiner2016} via a reduction from the Post Correspondence Problem (PCP) and HyperLTL can be embedded into HyperPDL-$\Delta$, we consider only restrictions of a fragment of  HyperPDL-$\Delta$ where quantified paths can be traversed linearly for the purpose of this section. 

\input{math/definitions/linearfragment}

For linear HyperPDL-$\Delta$, we generalise the semantics to arbitrary sets of traces $T$ instead of only those induced by KTS.
More concretely, for a fixed set of traces $T$, we let the assignment functions $\Pi$ map variables to traces in $T$ instead of paths of a structure. 
Such a function $\Pi$ is called a \textit{trace assignment}; the set of all trace assignments is called $\textit{TA}$. 
Furthermore, we let all quantifiers range over $T$ instead of $\mathit{Paths}(\mathcal{T}, \Pi(\epsilon)(0))$. 
We further define $T \models \varphi$ to hold for a linear HyperPDL-$\Delta$ formula $\varphi$ iff $\{\} \models \varphi$ holds for the trace assignment $\{\}$ with empty domain over $T$.
For linear HyperPDL-$\Delta$ formulas and using $T = \mathit{Traces}(\mathcal{T}, s_0)$, this definition coincides with the definition in \autoref{section:HyperPDL}.
We then call a linear HyperPDL-$\Delta$ formula $\varphi$ \textit{satisfiable} iff there is a non-empty set of traces $T$ such that $T \models \varphi$ holds.
The satisfiability problem for linear HyperPDL-$\Delta$ is to check whether a linear HyperPDL-$\Delta$ formula $\varphi$ is satisfiable.
For full linear HyperPDL-$\Delta$, we obtain undecidability via a reduction from HyperLTL satisfiability \cite{Finkbeiner2016}:

\input{math/theorems/satisfiability}

However, since the encoding relies on the availability of arbitrary combinations of quantifiers, a natural question is whether fragments of linear HyperPDL-$\Delta$ with restricted quantifier combinations yield a decidable satisfiability problem.
For HyperLTL, several such restrictions were considered \cite{Finkbeiner2016}. 
In \cite{Finkbeiner2016}, the satisfiability problem for the restricted fragments in HyperLTL is solved via the transformation of a HyperLTL formula to an equisatisfiable LTL formula and solving the satisfiability problem for the latter.
Since there is no apparent connection of this form  between HyperPDL-$\Delta$ and LTL, we use a similar type of translation, but instead translate our formulas into suitable ABA and check those for emptiness.
In all cases, the lower complexity bound can be obtained via reduction from the satisfiability problem of the corresponding fragment of HyperLTL.

\input{math/theorems/forallsatisfiability}

Similarly, we can show that the $\exists^*$ fragment is also PSPACE-complete.

\input{math/theorems/existssatisfiability}

For the $\exists^* \forall^*$-fragment, we eliminate the universal quantifiers by taking the variables bound by existential quantifiers and replacing the variables bound by universal quantifiers by all possible combinations of them. Unlike the previous two fragments, this increases the complexity beyond PSPACE.
\input{math/theorems/existsforallsatisfiability}

As in \cite{Finkbeiner2016}, from \autoref{theorem:existsforall}, we obtain that it is an EXPSPACE-complete problem to decide whether one uncritical HyperPDL-$\Delta$ formula is implied by another. In particular, the uncritical fragment includes properties like variants of observational determinism enriched by regular predicates. 

%% file: math/definitions/linearfragment.tex
\begin{definition}
	A linear HyperPDL-$\Delta$ formula is of the form $Q_1 \pi_1 \dots Q_n \pi_n . \psi$ for $Q_i \in \{\exists, \forall\}$ and $\psi$ a quantifier-free formula.
	A linear HyperPDL-$\Delta$ formula $\varphi$ is an $\mathcal \forall^*$-formula if $Q_i = \forall$ for all $ 1 \leq i \leq n$.
	The $\exists^*$-fragment and the $\exists^*\forall^*$-fragment are defined analogously.
\end{definition}

%% file: math/theorems/satisfiability.tex
\begin{theorem}
	The satisfiability problem for linear HyperPDL-$\Delta$ is undecidable. 
\end{theorem}

%% file: math/theorems/forallsatisfiability.tex
\begin{theorem}\label{forallsatisfiability}
	The satisfiability problem for the $\forall^*$-fragment of HyperPDL-$\Delta$ is PSPACE-complete.
\end{theorem}

\input{math/proofs/forallsatisfiability}

%% file: math/proofs/forallsatisfiability.tex
\begin{proof}(Sketch)
	Let $\varphi$ be a $\forall^*$-fragment formula, i.e. $\varphi \equiv \forall \pi_1 \dots \forall \pi_n . \psi$ for a quantifier-free formula $\psi$.
	We manipulate $\psi$ by substituting $\pi_1,...,\pi_n$ with a single fresh variable $\pi$ and obtain a formula $\psi'$.
	This is done by (i) replacing every occurence of an atomic proposition $a_{\pi_i}$ with $a_{\pi}$ and (ii) compressing each tuple of atomic programs $\tau$ into a program $\alpha$ with only $1$-tuples.
	The compression discriminates three cases.
	If $\tau$ only consists of wildcard programs, i.e. $\tau = \bullet$, then it is compressed to $(\cdot)$.
	If $\tau$ is composed from the set $\{\sigma,\cdot\}$ for some atomic program $\sigma$, then it is compressed to $(\sigma)$.
	Otherwise, i.e. if $\tau$ is composed of distinct non-wildcard atomic programs, it is compressed to $\false ? \cdot (\cdot)$.
	For example, for $\varphi \equiv \forall \pi_1 \forall \pi_2 \langle (\cdot,\sigma)\rangle a_{\pi_1} \land [\bullet + (\sigma,\sigma')] \neg a_{\pi_2}$, $\psi'$ is given by $\langle (\sigma) \rangle a_{\pi} \land [ (\cdot) + \false ? \cdot (\cdot)]\neg a_{\pi}$.
	Let $\mathcal{A}$ be the ABA for $\psi'$ as described in \autoref{Section:ModelChecking}.
	We can test $\mathcal{A}$ for emptiness in PSPACE \cite{Lange2016}. $\mathcal{A}$ is non-empty iff $\varphi$ is satisfiable:  any word $w$ accepted by $\mathcal{A}$ gives rise to the trace set $\{w\}$ satisfying $\varphi$. Analogously, a non-empty trace set $T$ with  $T \models \varphi$ can be used to obtain a word accepted by $\mathcal{A}$ since $\mathcal{A}$ accepts all traces satisfying $\psi'$ and we can thus pick any trace $t$ in $T$ as a witness. 
	
	The lower bound directly follows by a reduction from the satisfiability problem for the $\exists^* \forall^*$ fragment of HyperLTL \cite{Finkbeiner2016}.
\end{proof}

%% file: math/theorems/existssatisfiability.tex
\begin{theorem}\label{Theorem:ExistsSatisfiability}
	The satisfiability problem for the $\exists^*$-fragment of HyperPDL-$\Delta$ is PSPACE-complete.
\end{theorem}

\input{math/proofs/existssatisfiability}

%% file: math/proofs/existssatisfiability.tex
\begin{proof}(Sketch)
	Let $\varphi$ be a $\exists^*$-fragment formula, i.e. $\varphi \equiv \exists \pi_1 \dots \exists \pi_n . \psi$ for a quantifier-free formula $\psi$.
	We construct the alternating Büchi automaton $\mathcal{A}$ for $\psi$.
	Non-emptiness of $\mathcal{A}$ and satisfiability of $\varphi$ are  equivalent: every trace set $T$ fulfilling $\varphi$ contains  traces $t_1 \dots t_n$ fulfilling $\psi$ and these give rise to a word accepted by  $\mathcal{A}$. On the other hand, if $\mathcal{L}(\mathcal{A})$ is non-empty, the trace set $T$ induced by an arbitrary $w \in \mathcal{L}(\mathcal{A})$ fulfills $\varphi$.
	
	The lower bound follows straightforwardly by a reduction from the satisfiability problem for the $\exists^*$ fragment of HyperLTL \cite{Finkbeiner2016}.
\end{proof}

%% file: math/theorems/existsforallsatisfiability.tex
\begin{theorem}\label{theorem:existsforall}
	The satisfiability problem is EXPSPACE-complete for the $\exists^* \forall^*$-fragment of HyperPDL-$\Delta$.
\end{theorem}

\input{math/proofs/existsforallsatisfiability}

%% file: math/proofs/existsforallsatisfiability.tex
\begin{proof}(Sketch)
	Let $\varphi$ be a $\exists^* \forall^*$-formula, i.e. $\varphi \equiv \exists \pi_1 \dots \exists \pi_n  \forall \pi_1' \dots \forall \pi_m' . \psi$ for a quantifier-free formula $\psi$.
	For a formula $\psi$ and path variables $\pi, \pi'$, we define the substitution $\psi[\pi/\pi']$ to be the variant of $\psi$ in which all occurences of $\pi'$ have been replaced by $\pi$ similar to the substitution in the proof of \autoref{forallsatisfiability}.
	The main difference here is that only two instead of $n$ path variables are compressed into one.
	Thus only two instead of all atomic programs have to be considered when determining the replacement of a tuple.
	Let $\varphi' \equiv\exists \pi_1 \dots \exists \pi_n . \bigwedge_{j_1 = 1}^n   \dots  \bigwedge_{j_m = 1}^n \psi[\pi_{j_1} / \pi_1']\dots[\pi_{j_m} / \pi'_m].$ 
	For example, for $\varphi \equiv \exists \pi_1 . \exists \pi_2 . \forall \pi'_1 .$ $\langle (\cdot,\cdot,\sigma)\rangle a_{\pi_1} \land \neg a_{\pi_2} \land a_{\pi'_1}$, $\varphi' = \exists \pi_1 . \exists \pi_2 . (\langle (\sigma,\cdot)\rangle a_{\pi_1} \land \neg a_{\pi_2} \land a_{\pi_1}) \land (\langle (\cdot,\sigma)\rangle a_{\pi_1} \land \neg a_{\pi_2} \land a_{\pi_2})$.
	$\varphi'$ is an $\exists^*$ formula and is equisatisfiable to $\varphi$: any trace assignment satisfying $\varphi$ naturally induces a model of $\varphi'$.
	For the reverse direction, assume $\Pi \models \varphi'$.
	$\varphi'$ contains all possible combinations of assignments for the variables $\pi'_1 \dots \pi'_m$ with traces chosen for the existentially quantified variables $\pi_1 \dots \pi_n$.
	Then $T = \{\Pi(\pi_i) \mid 1 \leq i \leq n\} \models \varphi$. $\varphi'$ is constructible in EXPTIME.
	Therefore the satisfiability check is possible in EXPSPACE due to \autoref{Theorem:ExistsSatisfiability}.
	
	The lower bound easily follows by a reduction from the satisfiability problem for the $\exists^* \forall^*$ fragment of HyperLTL \cite{Finkbeiner2016}.
\end{proof}

%% file: sections/main/expressivity.tex
\section{Expressivity Results}\label{section:expressivity}
\label{expressivity}


As mentioned in the introduction, a desirable property of temporal logics is the ability to specify arbitrary $\omega$-regular properties.
We show that HyperPDL-$\Delta$ indeed has this property.

\input{math/theorems/omegaregular}

It follows from this theorem that HyperPDL-$\Delta$ can express an infinitary version of the regular hyperlanguages recently proposed in \cite{Bonakdarpour2020}.


We now compare our logic to other hyperlogics.
For this purpose we introduce a logic 
that adds to HyperCTL$^*$ the ability to quantify over atomic propositions.

\input{math/defs/QPTL}
\input{math/theorems/QPTL}
\input{math/proofs/QPTL}

By the results of \cite{Coenen2019}, we obtain that linear HyperPDL-$\Delta$ is stricly less expressive than S1S[E], which, as mentioned, has an undecidable model checking problem.
We leave a more precise localisation of linear and unrestricted HyperPDL-$\Delta$ in the hierarchies of hyperlogics from \cite{Coenen2019} for future work.
This includes comparisons with FO[$<$,E] and MPL[E] and an answer to the question if the second inequalities from the claims in \autoref{qptl} are indeed strict.

%% file: math/theorems/omegaregular.tex
\begin{theorem}\label{omegaregular}
Let $\Pi$ be a trace assignment and $\pi_1 \dots \pi_n$ be the variables bound by $\Pi$.
Let $\nu_{AP}: \textit{TA} \to ((2^{AP})^n \times \Sigma^n)^{\omega}$ be the analog of $\nu$ for trace assignments.
For a given $\omega$-regular language $\mathcal{L}$ over $(2^{AP})^n \times \Sigma^n$, there is a quantifier-free HyperPDL$-\Delta$ formula $\varphi$ with path variables $\pi_1 \dots \pi_n$ such that $\Pi \models \varphi$ iff $\nu_{AP}(\Pi) \in \mathcal{L}$.
\end{theorem}

\input{math/proofs/omegaregular}

%% file: math/proofs/omegaregular.tex
\begin{proof}	
	Let $\mathcal{L}$ be an $\omega$-regular language over $(2^{AP})^n \times \Sigma^n$.
	It is well-known \cite{Lange2016} that $\mathcal{L} = \bigcup_{i=1}^k \mathcal{L}_{i,0} \mathcal{L}_{i,1}^{\omega}$ holds for some regular languages $\mathcal{L}_{i,0}$, $\mathcal{L}_{i,1}$.
	Let $r_{i,j}$ be a regular expression for $\mathcal{L}_{i,j}$.
	Every symbol in $r_{i,j}$ has the form $((P_1,...,P_n),\tau)$ for $P_k \subseteq AP$ and $\tau \in \Sigma^n$.
	Let $\alpha_{i,j}$ be the regular expression obtained by replacing each such symbol in $r_{i,j}$ by $(\bigwedge_{l = 1}^{n}\bigwedge_{a \in P_l} a_{\pi_l} \land \bigwedge_{a \not\in P_l} \lnot a_{\pi_l})? \cdot \tau$.
	Then $\varphi \equiv \bigvee_{i=1}^k \langle \alpha_{i,0} \rangle \Delta \alpha_{i,1}$ yields the desired formula.
\end{proof}

%% file: math/defs/QPTL.tex
\begin{definition}[\cite{Coenen2019}]
The logic HyperQCTL$^*$ is obtained by adding to the syntax of HyperCTL$^*$ the rules $\varphi ::= q \mid \exists q . \varphi$ and to the semantics the rules
\begin{align*}
  \Pi \models_{\mathcal{T}}  q  &\quad\textit{ iff }\quad  q \in \Pi(\pi_q)(0) \\
  \Pi \models_{\mathcal{T}} \exists q . \varphi &\quad\textit{ iff }\quad \exists t \in (2^{\{q\}})^{\omega}.\Pi[\pi_q \rightarrow t] \models_{\mathcal{T}} \varphi
\end{align*}
The sub-logic HyperQPTL consists of the HyperQCTL$^*$ formulas where both path quantifiers $Q \pi. \varphi$ and propositional quantifiers $Q q. \varphi$ only occur at the front of the formula.
\end{definition} 

%% file: math/theorems/QPTL.tex
\begin{theorem}\label{qptl}
\begin{enumerate}
\item HyperCTL$^* <$  HyperPDL-$\Delta \leq$ HyperQCTL$^*$
\item HyperLTL $<$ Linear HyperPDL-$\Delta \leq $ HyperQPTL
\end{enumerate}
\end{theorem}

%% file: math/proofs/QPTL.tex
\begin{proof}
Part one of both claims is straightforward:
Embedding HyperLTL and HyperCTL$^*$ into (linear) HyperPDL-$\Delta$ works as described in \autoref{section:HyperPDL}.
An embedding in the other direction is impossible due to the inability of HyperLTL and HyperCTL$^*$ to express arbitrary $\omega$-regular properties \cite{Rabe2016}.

For the second part of the first claim, we observe that the semantics of HyperPDL-$\Delta$ can straightforwardly be encoded in MSO[E], which is equally expressive as HyperQCTL$^*$ by \cite{Coenen2019}.
The last claim can be shown by a direct translation via Büchi automata:
A quantifier-free HyperPDL-$\Delta$ formula can be translated into a Büchi Automaton as described in \autoref{Section:ModelChecking}.
By \cite{Kesten2002}, there is a QPTL formula for that automaton, where the quantifiers can be reattached to yield the desired formula.
\end{proof}

%% file: sections/conclusion/conclusion.tex
\section{Conclusion}\label{section:conclusion}
We introduced the logic HyperPDL-$\Delta$ as a variant of Propositional Dynamic Logic for hyperproperties that can express all $\omega$-regular properties.
Our model checking algorithm  has the same complexity as model checking HyperCTL$^*$, despite the increased expressive power.
Finally, we showed that satisfiability checking for certain fragments has the same complexity as for structurally similar, but less expressive fragments of HyperLTL.

Future work includes implementing a model checker for HyperPDL-$\Delta$.
It would also be interesting to explore alternative model checking and satisfiability testing algorithms for subfragments of HyperPDL-$\Delta$, possibly by exploiting classical techniques for PDL. 

%% file: sections/conclusion/appendix.tex
\clearpage
\appendix

\section{Detailed Proofs from \autoref{Section:ModelChecking}}\label{appendix:complexitylemma}
	
\input{math/proofs/constructionlemma}

\input{math/proofs/construction}

\input{math/proofs/complexitylemma}

\input{math/proofs/hardness}

\section{Missing proofs from \autoref{section:expressivity}}

\input{math/proofs/qptltranslation}

\section{Alternative Construction for the Transition Function of $\mathcal{A}_{\varphi}$}\label{appendix:alphaalt}

\input{math/constructions/alphaalt}

%% file: math/proofs/constructionlemma.tex
\begin{proof}[Proof of \autoref{constructionlemma}]
	By structural induction over the structure of $\alpha$.
	
	The fact that $i$ and $k$ are even ensures that the path assignment $\Pi[i,k]$ starts and ends with states rather than atomic programs.
	
	\boldmath$\alpha = \tau$\unboldmath: Assume $(\Pi,i,k) \in R(\alpha)$. 
	By definition, this implies $\tau|_l = \cdot$ or $\tau_{i+1}|_l = \tau|_l$ for all $l$ (*) and $k = i + 2$.
	We construct a state sequence in $M_{\alpha}$ satisfying (i) to (iv).
	Such a state sequence can only consist of two states and only a single symbol can be read due to the length restriction.
	The state sequence $q_0 q_3$ with empty sets $X_0,X_1$ satisfies this constraint as well as (i) and (ii) by construction, (iii) by (*) and (iv) trivially.
	
	Let there be a state sequence in $M_{\alpha}$ satisfying conditions (i) to (iv).
	By construction and conditions (i) to (iii), we get that the state sequence can only consist of $q_0$ and $q_3$, thus $k = i+2$ and $\tau_{i+1}$ is the only symbol read.
	Additionally, we get that for all $l$ $\tau|_l = \cdot$ or $\tau|_l = \tau_{i+1}|_l$ which establishes $(\Pi,i,k) \in R(\alpha)$ by definition.
	
	\boldmath$\alpha = \varepsilon$\unboldmath: Assume, that $(\Pi,i,k) \in R(\alpha)$ holds.
	This implies, by definition, that $k = i$ and the state sequence we are looking for can only consist of a single state.
	Indeed, the singleton state sequence $q_0$ with an empty set $X_0$ satisfies conditions (i) to (iv): it is the initial state and the final state is $\varepsilon$-reachable which establishes (i) and (ii), whereas the remaining conditions (iii) and (iv) are established trivially.
	
	We assume there is a state sequence in $M_{\alpha}$ satisfying (i) to (iv).
	Since there are only $\varepsilon$-transitions in $M_{\alpha}$, it can only consist of a single state, the initial state $q_0$.
	Therefore, since $m = 0$, we have $i = k$, which by definition establishes $(\Pi,i,k) \in R(\alpha)$.
	
	\boldmath$\alpha = \alpha_{1} + \alpha_{2}$\unboldmath: First, assume that $(\Pi,i,k) \in R(\alpha)$ holds.
	By definition, this implies either $(\Pi,i,k) \in R(\alpha_{1})$ or $(\Pi,i,k) \in R(\alpha_{2})$.
	We distinguish two cases where the first (resp. the second) condition holds and consider $(\Pi,i,k) \in R(\alpha_{1})$.
	The other case is analogous.
	By the induction hypothesis, we get a state sequence $q_{0,\alpha_{1}} ... q_{m,\alpha_{1}}$ with $m = \frac{k-i}{2}$ and sets $X_0,...,X_m$ satisfying conditions (i) to (iv).
	We replace $q_{0,\alpha_{1}}$ by $q_0$ and $q_{m,\alpha}$ by $q_f$ to obtain a state sequence of the same length in $M_{\alpha}$.
	Condition (i) is already established.
	$\varepsilon$-transitions from $q_0$ to $q_{0,\alpha_{1}}$ and from $q_{m,\alpha}$ to $q_f$ make sure that (ii) and (iii) still hold.
	Finally, since $q_{0,\alpha_{1}} ... q_{m,\alpha_{1}}$ satisfies (iv) and neither $q_0$ nor $q_f$ are marked, (iv) is established as well with the same sets $X_i$.
	
	Now, assume there is a state sequence with sets $X_0,...,X_m$ in $M_{\alpha}$ satisfying (i) to (iv).
	Since, by construction, $M_{\alpha_{1}}$ and $M_{\alpha_{2}}$ are unconnected subautomata of $M_{\alpha}$, the state sequence has to move from $q_0$ to $q_f$ either through a state sequence in $M_{\alpha_{1}}$ or in $M_{\alpha_{2}}$.
	We differentiate the two cases and assume the first one.
	The second case is analogous.
	We replace $q_0$ with $q_{0,\alpha_1}$ and $q_f$ with $q_{f,\alpha_1}$ in the full sequence.
	Since states in $M_{\alpha_1}$ are only reachable from $q_0$ via $q_{0,\alpha_1}$ and $q_f$ is only reachable from states in $M_{\alpha_1}$ via $q_{f,\alpha_1}$, conditions (i) to (iii) still hold for this state sequence in $M_{\alpha_1}$.
	Condition (iv) is established with the same sets $X_i$ by the property of our construction that initial and final states can never be marked.
	By induction hypothesis, we get that $(\Pi,i,k) \in R(\alpha_{1})$ which satisfies the requirements for $(\Pi,i,k) \in R(\alpha)$.
	
	\boldmath$\alpha = \alpha_{1} \cdot \alpha_{2}$\unboldmath: Assume that $(\Pi,i,k) \in R(\alpha)$.
	By the definition of the semantics of $\alpha$, there is $j$ with $i \leq j \leq k$ such that $(\Pi,i,j) \in R(\alpha_{1})$ and $(\Pi,j,k) \in R(\alpha_{2})$.
	Using the induction hypothesis twice, we get a state sequence $q_{1,0} ... q_{1,m_1}$ with $m_1 = \frac{j-i}{2}$ and sets $X_{1,0},...,X_{1,m_1}$ in $M_{\alpha_1}$ and a state sequence $q_{2,0} ... q_{2,m_2}$ with $m_2 = \frac{k-j}{2}$ and sets $X_{2,0},...,X_{2,m_2}$ in $M_{\alpha_2}$, both satisfying conditions (i) to (iv).
	Since $q_{f,1}$ is $\varepsilon$-reachable from $q_{1,m_1}$ by (ii) and a new $\varepsilon$-transition was introduced from $q_{f,1}$ to $q_{0,2}$, $q_{2,0}$ is $\varepsilon$-reachable from $q_{1,m_1}$ as well.
	Thus, $q_{1,m_1}$ inherits all $\varepsilon$- and $\tau$-reachable states from $q_{2,0}$ and we can concatenate the two state sequences, leaving out $q_{2,0}$.
	One can easily see that this is indeed a proper state sequence in $M_{\alpha}$ fulfilling (i) to (iii).
	Since the merged state sequences consist of $m_1 + 1$ resp. $m_2 + 1$ states and one state was removed, the new state sequence has $m + 1 = m_1 + 1 + m_2 + 1 - 1$ states for $m = m_1 + m_2 = \frac{j-i}{2} + \frac{k-j}{2} = \frac{k-i}{2}$.
	Condition (iv) is established with condition (iv) of the initial state sequences and by recognising that the removed state was not marked with a formula.
	Therefore one only has to merge the sets $X_{1,m_1}$ and $X_{2,0}$.
	
	For the reverse direction, assume that there is a state sequence with sets $X_0,...,X_m$ in $M_{\alpha}$ satisfying (i) to (iv).
	Since the final state of $M_{\alpha}$ is in $Q_2$ and the initial state of $M_{\alpha}$ is in $Q_1$, the state sequence has to transition from $Q_1$ to $Q_2$ at some point.
	Due to the form of $\rho$, this can only happen once by taking the $\varepsilon$-transition from $q_{f,1}$ to $q_{0,2}$.
	By inserting the state $q_{0,2}$ into the sequence where the $\varepsilon$-transition was taken, we obtain two state sequences $q_{1,0} ... q_{1,m_1}$ and $q_{2,0} ... q_{2,m_2}$ in the subautomata $M_{\alpha_{1}}$ and $M_{\alpha_{2}}$ respectively.
	Corresponding sets $X_{1,0},...,X_{1,m_1}$ and $X_{2,0},...,X_{2,m_2}$ are obtained by splitting the set $X_{m_1}$ in an appropriate way.
	One can easily see that these state sequences fulfill conditions (i) to (iii).
	Their length is $m_1 + 1$ resp. $m_2 + 1$ for some $m_1,m_2$ with $m_1 + m_2 = \frac{k-i}{2}$.
	Thus, there is a $j$ such that $m_1 = \frac{j-i}{2}$ and $m_2 = \frac{k-j}{2}$.
	As condition (iv) for the original state sequence implies condition (iv) for the new state sequences (with indices shifted by $j-i$ for the second path), this enables us to use the induction hypothesis twice to directly obtain the semantics definition of $(\Pi,i,k) \in R(\alpha)$.
	
	\boldmath$\alpha = \alpha_{1}^{*}$\unboldmath: Assume $(\Pi,i,k) \in R(\alpha)$.
	Using the semantics definition, we obtain indices $i = j_0 \leq ... \leq j_l = k$ with $(\Pi,j_p,j_{p+1}) \in R(\alpha_1)$ for $0 \leq p < l$.
	For $l = 0$, $i = k$ has to hold and the singleton state sequence $q_0$ with the empty set $X_0$ establishes our claim.
	If $l > 0$, one can use the induction hypothesis on each pair $j_p,j_{p+1}$ to obtain a state sequence in $M_{\alpha_1}$.
	For $l = 1$, the state sequence and sets in $M_{\alpha_{1}}$ trivially translates to a state sequence and sets in $M_{\alpha}$.
	For $l > 1$, we concatenate every two consecutive state sequences just as in the previous case $\alpha = \alpha_{1} \alpha_{2}$.
	This leaves us with the desired state sequence in $M_{\alpha}$ with the same arguments as given there.
	
	Now assume that there is a state sequence with sets $X_0,...,X_m$ in $M_{\alpha}$ fulfilling conditions (i) to (iv).
	This state sequence can take the $\varepsilon$-transition between $q_{f}$ and $q_{0}$ arbitrarily many times.
	If it does not do so, it is either the state sequence $q_0$ or a state sequence in $M_{\alpha_{1}}$.
	The singleton state sequence $q_0$ with $m = 0$ implies $i = k$ and a single index $j = i = k$ fulfills the semantics definition of $(\Pi,i,k) \in R(\alpha)$.
	For a state sequence in $M_{\alpha_1}$, the claim is directly established by the induction hypothesis for two indices $j_1 = i, j_2 = k$.
	Any other state sequence can be divided into sub-sequences by inserting states $q_{0,1}$ and splitting sets just like in the previous case $\alpha = \alpha_{1} \cdot \alpha_{2}$.
	Conditions (i) to (iii) are easily established for these state sequences.
	Condition (iv) is obtained by introducing appropriate indices $j_1,...,j_l$ as substitutions for $i$ and $k$ for each state sequence.
	Using the induction hypothesis on each of the state sequences gives us the semantics definition of $(\Pi,i,k) \in R(\alpha)$ with indices $j_1,...,j_l$.
	
	\boldmath$\alpha = \varphi?$\unboldmath: Assume that $(\Pi,i,k) \in R(\alpha)$.
	By the definition of semantics, we obtain $i = k$ and $\Pi[i, \infty] \models_{\mathcal{T}} \varphi$.
	This directly implies that the state sequence $q_0$ with the set $X_0 = \{\psi\}$ in $M_{\alpha}$ fulfills conditions (i) to (iv).
	
	For the other direction, assume that there is a state sequence with sets $X_0,...,X_m$ in $M_{\alpha}$ fulfilling (i) to (iv).
	Since there are only $\varepsilon$-transitions, the state sequence can only be $q_0$.
	Therefore, we know that $k = i$ (by $m = 0$) and that $\Pi[i,\infty] \models_{\mathcal{T}} \varphi$ (by (iv) and the fact that $q_2$ can only be reached by passing through $q_1$, thus $\psi \in X_0$), which is exactly the definition for $(\Pi,i,k) \in R(\alpha)$.
\end{proof}

%% file: math/proofs/construction.tex
\begin{proof}[Proof of \autoref{construction}]
	By structural induction over the structure of $\varphi$.
	In each case, let $\Pi$ be a path assignment with $\nu(\Pi) = (\mathpzc{s}_0, \tau_0) (\mathpzc{s}_1, \tau_1) ...$ where $\mathpzc{s}_j = (s_0^j,...,s_n^j)$ and $\tau_j = (\sigma_0^j,...,\sigma_n^j)$.
	
	\boldmath$\varphi = a_{\pi_k}$\unboldmath: 
	By construction, $\mathcal{L}(\mathcal{A}_{\varphi}) = \{(\mathpzc{s}_0', \tau_0') (\mathpzc{s}_1', \tau_1') ... \in (S^n \times \Sigma^n)^{\omega} | a \in L(\mathpzc{s}_0' |_{k})\}$.

	On the one hand, assume $\nu(\Pi) \in \mathcal{L}(\mathcal{A}_{\varphi})$.
	Then, by the above characterisation of $\mathcal{L}(\mathcal{A}_{\varphi})$, we have $a \in L(s_k^0)$ which by definition of $\nu$ translates to $a \in L(\Pi(\pi_k)(0))$.
	Therefore, $\Pi \models_{\mathcal{T}} \varphi$.

	On the other hand, assume $\Pi \models_{\mathcal{T}} \varphi$.
	By the definition of $\varphi$'s semantics, this implies $a \in L(\Pi(\pi_k)(0))$ which by definition of $\nu$ translates to $a \in L(s_k^0)$.
	Therefore, we have $\nu(\Pi) \in \mathcal{L}(\mathcal{A}_{\varphi})$ by the above characterisation of $\mathcal{L}(\mathcal{A}_{\varphi})$.
	
	\boldmath$\varphi = \lnot a_{\pi_k}$\unboldmath: 
	analogous to $a_{\pi_k}$.
	
	\boldmath$\varphi = \varphi_1 \land \varphi_2$\unboldmath:
	By construction, we have $\mathcal{L}(\mathcal{A}_{\varphi}) = \mathcal{L}(\mathcal{A}_{\varphi_1}) \cap \mathcal{L}(\mathcal{A}_{\varphi_2})$.
	
	On the one hand, assume that $\nu(\Pi) \in \mathcal{L}(\mathcal{A}_{\varphi})$.
	By the above characterisation of $\mathcal{L}(\mathcal{A}_{\varphi})$, we have that $\nu(\Pi) \in \mathcal{L}(\mathcal{A}_{\varphi_1})$ and $\nu(\Pi) \in \mathcal{L}(\mathcal{A}_{\varphi_2})$.
	Using the induction hypothesis twice, we get that $\mathcal{A}_{\varphi_{i}}$ is $\mathcal{T}$-equivalent to $\varphi_{i}$ for $i = 1,2$.
	Therefore, $\Pi \models_{\mathcal{T}} \varphi_1$ and $\Pi \models_{\mathcal{T}} \varphi_2$ which implies $\Pi \models_{\mathcal{T}}$ by $\varphi$'s semantics.
	
	On the other hand, assume that $\Pi \models_{\mathcal{T}} \varphi$.
	By the definition of $\varphi$'s semantics we have $\Pi \models_{\mathcal{T}} \varphi_1$ and $\Pi \models_{\mathcal{T}} \varphi_2$.
	Using the induction hypothesis twice we get that $\mathcal{A}_{\varphi_{i}}$ is $\mathcal{T}$-equivalent to $\varphi_{i}$ for $i = 1,2$.
	Therefore, we have that $\nu(\Pi) \in \mathcal{L}(\mathcal{A}_{\varphi_1})$ and $\nu(\Pi) \in \mathcal{L}(\mathcal{A}_{\varphi_2})$ which by the above characterisation of $\mathcal{L}(\mathcal{A}_{\varphi})$ implies $\nu(\Pi) \in \mathcal{L}(\mathcal{A}_{\varphi})$.
	
	\boldmath$\varphi = \varphi_1 \lor \varphi_2$\unboldmath:
	By construction, we see that $\mathcal{L}(\mathcal{A}_{\varphi}) = \mathcal{L}(\mathcal{A}_{\varphi_1}) \cup \mathcal{L}(\mathcal{A}_{\varphi_2})$.
	
	On the one hand, assume that $\nu(\Pi) \in \mathcal{L}(\mathcal{A}_{\varphi})$.
	By the above characterisation of $\mathcal{L}(\mathcal{A}_{\varphi})$, we have that $\nu(\Pi) \in \mathcal{L}(\mathcal{A}_{\varphi_1})$ or $\nu(\Pi) \in \mathcal{L}(\mathcal{A}_{\varphi_2})$.
	We discriminate the cases and assume the former.
	The latter case is analogous.
	Then, by induction hypothesis, since $\mathcal{A}_{\varphi_{1}}$ is $\mathcal{T}$-equivalent to $\varphi_{1}$, we get $\Pi \models_{\mathcal{T}} \varphi_1$ which implies $\Pi \models_{\mathcal{T}} \varphi$ by the definition of $\varphi$'s semantics.
	
	On the other hand, assume that $\Pi \models_{\mathcal{T}} \varphi$.
	By the definition of $\varphi$'s semantics, we get $\Pi \models_{\mathcal{T}} \varphi_1$ or $\Pi \models_{\mathcal{T}} \varphi_2$.
	We discriminate the cases and assume the former.
	The other case is analogous.
	Then, by induction hypothesis, since $\mathcal{A}_{\varphi_{1}}$ is $\mathcal{T}$-equivalent to $\varphi_{1}$, we get $\nu(\Pi) \in \mathcal{L}(\mathcal{A}_{\varphi_1})$ which shows the claim by the above characterisation of $\mathcal{L}(\mathcal{A}_{\varphi})$.
	
	\boldmath$\varphi = \langle \alpha \rangle \varphi_1$\unboldmath:
	By construction, $\mathcal{L}(\mathcal{A}_{\varphi})$ can be characterised as follows:
	$w = (\mathpzc{s}_0', \tau_0') (\mathpzc{s}_1', \tau_1') ... \in \mathcal{L}(\mathcal{A}_{\varphi})$ iff there is $k$ such that: \\ (i) there is a path $q_0 q_1 ... q_k$ from the initial state $q_{0,\alpha}$ to a $q_k$ with $q_k \xRightarrow{\varepsilon}_{X} q_{f,\alpha}$ for some $X$ obtained by reading $\tau_0' \tau_1' ... \tau_{k-1}'$\\
	(ii) for every state $q_j$ with $j < k$ on the path from (i), $\psi_l \in X$ with $q_j \xRightarrow{\tau}_X q_{j+1}$ implies $(\mathpzc{s}_j', \tau_j') (\mathpzc{s}_{j+1}', \tau_{j+1}') ... \in \mathcal{L}(\mathcal{A}_{\psi_l})$ \\
	(iii) $\psi_l \in X$ with $q_k \xRightarrow{\varepsilon}_X q_{f,\alpha}$ implies $(\mathpzc{s}_k', \tau_k') (\mathpzc{s}_{k+1}', \tau_{k+1}') ... \in \mathcal{L}(\mathcal{A}_{\psi_l})$ \\
	(iv) $(\mathpzc{s}_k', \tau_k') (\mathpzc{s}_{k+1}', \tau_{k+1}') ... \in \mathcal{L}(\mathcal{A}_{\varphi_1})$
	
	On the one hand, assume that $\nu(\Pi) \in \mathcal{L}(\mathcal{A}_{\varphi})$.
	By using the induction hypothesis on all subautomata of $\mathcal{A}_{\varphi}$, we get that they all are $\mathcal{T}$-equivalent to their respective formula.
	With the above characterisation and the $\mathcal{T}$-equivalences, we get that there is a $k$ such that: \\
	(i) there is a path $q_0 q_1 ... q_k$ from $q_{0,\alpha}$ to a $q_k$ with $q_k \xRightarrow{\varepsilon}_X q_{f,\alpha}$ for some $X$ obtained by reading $\tau_0 ... \tau_{k-1}$ and \\
	(ii) for every state $q_j$ with $j < k$ on the path from (i), $\psi_l \in X$ with $q_j \xRightarrow{\tau}_X q_{j+1}$ implies $\Pi[j,\infty] \models_{\mathcal{T}} \psi_l$ \\
	(iii) $\psi_l \in X$ with $q_k \xRightarrow{\varepsilon}_X q_{f,\alpha}$ implies $\Pi[k,\infty] \models_{\mathcal{T}} \psi_l$ \\
	(iv) $\Pi[k,\infty] \models_{\mathcal{T}} \varphi_{1}$ \\
	Since the path in (i) is a state sequence in $M_{\alpha}$ and has appropriate length, we can use \autoref{constructionlemma} to obtain that $(\Pi,0,k) \in R(\alpha)$.
	This, together with (iv), implies $\Pi \models_{\mathcal{T}} \varphi$.
	
	On the other hand, assume that $\Pi \models_{\mathcal{T}} \varphi$.
	By using the induction hypothesis on all subformulas of $\varphi$, we get that they all are $\mathcal{T}$-equivalent to their respective automaton.
	The definition of $\varphi$'s semantics implies that there is a $k$ such that $(\Pi,0,k) \in R(\alpha)$ and $\Pi[k,\infty] \models_{\mathcal{T}} \varphi_{1}$.
	Using \autoref{constructionlemma}, we obtain a state sequence $q_0 ... q_m$ with $m = \frac{k}{2}$ and sets $X_0,...,X_m$ in $M_{\alpha}$ such that \\
	(i) $q_0$ is the initial state of $M_{\alpha}$, \\
	(ii) $q_m \xRightarrow{\varepsilon}_{X_m} q_{f,\alpha}$ for the final state $q_{f,\alpha}$ of $M_{\alpha}$, \\
	(iii) $q_l \xRightarrow{\tau_{2l+1}}_{X_l} q_{l+1}$ for all $l < m$, \\
	(iv) $\psi_i \in X_l$ implies $\Pi[2l,\infty] \models_{\mathcal{T}} \psi_i$ for all $l \leq m$. \\
	Using the $\mathcal{T}$-equivalence of $\varphi$'s subformulas with their respective automata, we can translate these conditions such that there is $k$ with: \\
	(i) there is a path $q_0 q_1 ... q_k$ from the initial state $q_{0,\alpha}$ to a $q_k$ with $q_k \xRightarrow{\varepsilon}_X q_{f,\alpha}$ for some $X$ obtained by reading $\tau_0 ... \tau_{k-1}$ \\
	(ii) for every state $q_j$ with $j < k$ on the path, $\psi_l \in X $ with $q_j \xRightarrow{\tau}_X q_{j+1}$ implies $\nu(\Pi[j,\infty]) \in \mathcal{L}(\mathcal{A}_{\psi_l})$ \\
	(iii) $\psi_l \in X$ with $q_k \xRightarrow{\varepsilon}_X q_{f,\alpha}$ implies $\nu(\Pi[k,\infty]) \in \mathcal{L}(\mathcal{A}_{\psi_l})$ \\
	(iv) $\nu(\Pi[k,\infty]) \in \mathcal{L}(\mathcal{A}_{\varphi_{1}})$ \\
	By the above definition of $\mathcal{L}(\mathcal{A}_{\varphi})$, this implies that $\nu(\Pi) \in \mathcal{L}(\mathcal{A}_{\varphi})$.
	
	\boldmath$\varphi = [ \alpha ] \varphi_1$\unboldmath:
	By construction, $\mathcal{L}(\mathcal{A}_{\varphi})$ satisfies $w = (\mathpzc{s}_0', \tau_0') (\mathpzc{s}_1', \tau_1') ... \in \mathcal{L}(\mathcal{A}_{\varphi})$ if and only if \\
	(*) for all paths $q_0 q_1 ... q_k$ with (i) $q_0 = q_{0,\alpha}$ and $q_k \xRightarrow{\varepsilon}_{X_k} q_{f,\alpha}$, (ii) $q_j \xRightarrow{\tau_j'}_{X_j} q_{j+1}$ for all $j < k$, and (iii) $\psi_l \not\in X_j$ or $(\mathpzc{s}_j', \tau_j') (\mathpzc{s}_{j+1}', \tau_{j+1}') ... \not\in \mathcal{L}(\mathcal{A}_{\lnot\psi_l})$ for all $j \leq k$ we have $(\mathpzc{s}_k', \tau_k') (\mathpzc{s}_{k+1}', \tau_{k+1}') ... \in \mathcal{L}(\mathcal{A}_{\varphi_1})$

	On the one hand, assume that $\nu(\Pi) \in \mathcal{L}(\mathcal{A}_{\varphi})$.
	By using the induction hypothesis on all subautomata of $\mathcal{A}_{\varphi}$, we conclude that they all are $\mathcal{T}$-equivalent to their respective formula.
	With the above characterisation of $\mathcal{L}(\mathcal{A}_{\varphi})$ and making use of $\mathcal{T}$-equivalences, we obtain that: \\
	(*) for all paths $q_0 q_1 ... q_k$ with (i) $q_0 = q_{0,\alpha}$ and $q_k \xRightarrow{\varepsilon}_{X_k} q_{f,\alpha}$ (ii) $q_j \xRightarrow{\tau_j}_{X_j} q_{j+1}$ for all $j < k$ (iii) $\psi_l \in X_j$ implies $\Pi[j,\infty] \not\models_{\mathcal{T}} \lnot\psi_l$ for all $j \leq k$ we have $\Pi[k,\infty] \models_{\mathcal{T}} \varphi_1$ \\
	Eliminating the double negation in (iii) of (*), we can use \autoref{constructionlemma} to obtain that for all $k$, $(\Pi,0,k) \in R(\alpha)$ implies $\Pi[k,\infty] \models_{\mathcal{T}} \varphi_{1}$, which establishes the claim.
	
	On the other hand, assume that $\Pi \models_{\mathcal{T}} \varphi$.
	By using the induction hypothesis on all subformulas of $\varphi$, we get that they all are $\mathcal{T}$-equivalent to their respective automaton.
	The definition of $\varphi$'s semantics implies that for all $k$, $(\Pi,0,k) \in R(\alpha)$ implies $\Pi[k,\infty] \models_{\mathcal{T}} \varphi$.
	Using \autoref{constructionlemma}, we can reformulate the premise of the implication to be the existence of a state sequence $q_0 ... q_k$ in $M_{\alpha}$ such that $q_0$ is the initial state, $q_j \xRightarrow{\tau_j}_{X_j} q_{j+1}$ for $j < k$, $q_k \xRightarrow{\varepsilon}_{X_k} q_{f,\alpha}$ for the final state $q_{f,\alpha}$ of $M_{\alpha}$ and for all $j \leq k$, $\psi_l \in X_j$ implies $\Pi[j,\infty] \models_{\mathcal{T}} \psi_l$.
	Making use of $\mathcal{T}$-equivalences and introducing a double negation, we get that: \\
	(*) for all paths $q_0 q_1 ... q_k$ with (i) $q_0 = q_{0,\alpha}$ and $q_k \xRightarrow{\varepsilon}_{X_k} q_{f,\alpha}$ (ii) $q_j \xRightarrow{\tau_j'}_{X_j} q_{j+1}$ for all $j < k$ (iii) $\psi_l \not\in X_j$ or $\nu(\Pi[j,\infty]) ... \not\in \mathcal{L}(\mathcal{A}_{\lnot\psi_l})$ for all $j \leq k$ we have $\nu(\Pi[k,\infty]) ... \in \mathcal{L}(\mathcal{A}_{\varphi_1})$\\
	This, by the above characterisation of $\mathcal{L}(\mathcal{A}_{\varphi})$ implies $\nu(\Pi) \in \mathcal{L}(\mathcal{A}_{\varphi})$.
	
	\boldmath$\varphi = \Delta \alpha$\unboldmath:
	By construction, $\mathcal{L}(\mathcal{A}_{\varphi})$ satisfies
	$w = (\mathpzc{s}_0', \tau_0') (\mathpzc{s}_1', \tau_1') ... \in \mathcal{L}(\mathcal{A}_{\varphi})$ if and only if \\
	(i) there are infinitely many indices $0 = k_1 \leq k_2 \leq ...$ such that for all $i$, $q_{k_i} = q_0$ and there is a path $q_{k_i} ... q_{k_{i+1}}$ induced by a state sequence $q_0 q_1 ... q_m$ with $q_0 = q_{0,\alpha}$, $q_j \xRightarrow{\tau_{k_i + j}'}_{X_j} q_{j+1}$ for $j < m$ and $q_m \xRightarrow{\varepsilon}_{X_m} q_{f,\alpha}$ and \\
	(ii) for each $0 \leq j \leq m$ in such a state sequence $\psi_{l} \in X_j$ implies \\ $(\mathpzc{s}_{k_i + j}', \tau_{k_i + j}') (\mathpzc{s}_{k_i + j+1}', \tau_{k_i + j+1}') ... \in \mathcal{L}(\mathcal{A}_{\psi_l})$
	
	On the one hand, assume that $\nu(\Pi) \in \mathcal{L}(\mathcal{A}_{\varphi})$.
	By using the induction hypothesis on all subautomata of $\mathcal{A}_{\varphi}$, we see that they all are $\mathcal{T}$-equivalent to their respective formula.
	The above characterisation of $\mathcal{L}(\mathcal{A}_{\varphi})$ implies that: \\
	(i) there are infinitely many indices $0 = k_1 \leq k_2 \leq ...$ such that for all $i$, $q_{k_i} = q_0$ and there is a path $q_{k_i} ... q_{k_{i+1}}$ induced by a state sequence $q_0 q_1 ... q_m$ with $q_0 = q_{0,\alpha}$, $q_j \xRightarrow{\tau_{k_i + j}'}_{X_j} q_{j+1}$ for $j < m$ and $q_m \xRightarrow{\varepsilon}_{X_m} q_{f,\alpha}$ and \\
	(ii) for each $0 \leq j \leq m$ in such a state sequence $\psi_{l} \in X_j$ implies $\nu(\Pi[k_i + j, \infty]) ... \in \mathcal{L}(\mathcal{A}_{\psi_l})$ \\
	Using $\mathcal{T}$-equivalences, we can translate (ii) to a form where \autoref{constructionlemma} is applicable to the paths in (i) and (ii) to directly obtain the semantics definition of $\Pi \models_{\mathcal{T}} \varphi$.
	
	On the other hand, assume that $\Pi \models_{\mathcal{T}} \varphi$.
	Using the induction hypothesis on all subformulas of $\varphi$, we get that they all are $\mathcal{T}$-equivalent to their respective automaton.
	The definition of $\varphi$'s semantics implies that there are infinitely many indices $0 = k_1 \leq k_2 \leq ...$ such that for all $i \geq 1$, $(\Pi,k_i,k_{i+1}) \in R(\alpha)$ holds.
	For each of these triples, we can use \autoref{constructionlemma} to receive a state sequence $q_{0} ... q_{m}$ with sets $X_0,...,X_m$ in $M_{\alpha}$ such that $q_0 = q_{o,\alpha}$, $q_{j} \xRightarrow{\tau_{k_i + j}}_{X_j} q_{j+1}$ for $j < m$ and $q_k \xRightarrow{\varepsilon}_{X_m} q_{f,\alpha}$ for the final state $q_{f,\alpha}$ of $M_{\alpha}$ where $\psi_l \in X_j$ implies $\Pi[k_i + j,\infty] \models_{\mathcal{T}} \psi_l$.
	By replacing the first and last state of each sequence with $q_0$ and using $\mathcal{T}$-equivalences, we directly obtain the above characterisation for $\nu(\Pi) \in \mathcal{L}(\mathcal{A}_{\varphi})$.
	
	\boldmath$\varphi = \lnot \Delta \alpha$\unboldmath:
	This case is directly implied by the induction hypothesis and \autoref{complementation}.
	
	\boldmath$\varphi = \exists \pi. \varphi_1$\unboldmath:
	Since $\mathcal{A}_{\varphi}$ is a Büchi automaton, its language satisfies $w \in \mathcal{L}(\mathcal{A}_{\varphi})$ iff there is a path $q_0 (q_1,s_1,\sigma_1)(q_2,s_2,\sigma_2)...$ in $\mathcal{A}_{\varphi}$ witnessing the acceptance of $w$.
	
	On the one hand, assume that $\nu(\Pi) \in \mathcal{L}(\mathcal{A}_{\varphi})$.
	By using the induction hypothesis on $\mathcal{A}_{\varphi_{1}}$, we get that it is $\mathcal{T}$-equivalent to $\varphi_{1}$.
	Let $q_0 (q_1,s_1,\sigma_1)(q_2,s_2,\sigma_2)...$ be the path in $\mathcal{A}_{\varphi}$ witnessing the acceptance of $\nu(\Pi)$ as in the above characterisation.
	Note that by construction of $\mathcal{A}_{\varphi}$, $p = s_0 \sigma_0 s_1 \sigma_1...$ for appropriate $s_0, \sigma_0$ is also a path of $\mathcal{T}$, starting at $\Pi(\epsilon) = \mathpzc{s}|_n$, or more formally, $p \in Paths(\mathcal{T},\Pi(\epsilon)(0))$.
	Let $\Pi'$ be the path assignment $\Pi[\pi_{n+1} \to p, \epsilon \to p]$.
	Then, the construction also makes sure that $\nu(\Pi') \in \mathcal{L}(\mathcal{A}_{\varphi_1})$, which implies that $\Pi' \models_{\mathcal{T}} \varphi_{1}$ by the use of a $\mathcal{T}$-equivalence.
	Therefore, we directly obtain $\Pi \models_{\mathcal{T}} \varphi$.
	
	On the other hand, assume that $\Pi \models_{\mathcal{T}} \varphi$.
	By the induction hypothesis, $\varphi_{1}$ is $\mathcal{T}$-equivalent to $\mathcal{A}_{\varphi_{1}}$.
	The definition of $\varphi$'s semantics implies that there is a path $p \in Paths(\mathcal{T},\Pi(\epsilon)(0))$ such that $\Pi' := \Pi[\pi \to p, \epsilon \to p] \models_{\mathcal{T}} \varphi_{1}$.
	Using $\mathcal{T}$-equivalence, we obtain that $\nu(\Pi') \in \mathcal{L}(\mathcal{A}_{\varphi_{1}})$.
	The fact that we use a dealternised version of $\mathcal{A}_{\varphi_{1}}$ equips us with a path $q_0 q_1 ...$ in $\mathcal{A}_{\varphi_{1}}$ witnessing the acceptance of $\nu(\Pi')$.
	Since $p = s_0 \sigma_0 s_1 \sigma_1 ...$ is a path in $\mathcal{T}$, starting at $\Pi(\epsilon) = \mathpzc{s}|_n$, we can construct a path $q_0 (q_1,s_1,\sigma_1)(q_2,s_2,\sigma_2)...$ in $\mathcal{A}_{\varphi}$ witnessing the acceptance of $\nu(\Pi)$ in $\mathcal{A}_{\varphi}$ which by the above characterisation implies $\nu(\Pi) \in \mathcal{L}(\mathcal{A}_{\varphi})$.
	
	\boldmath$\varphi = \lnot \exists \pi. \varphi_1$\unboldmath:
	This case is directly implied by the induction hypothesis, \autoref{dealternation} and \autoref{complementation}.
\end{proof}

%% file: math/proofs/complexitylemma.tex
\begin{proof}[Proof of \autoref{complexitylemma}]
	By induction on the criticality $k$.
	
	\textbf{Base case: }
	$\varphi$ has criticality $0$. 
	We show inductively that $\mathcal{A}_{\varphi}$ has size $2^{\mathcal{O}(p(n)+p'(log(m)))}$ in the size $n$ of $\varphi$ and the size $m$ of $\mathcal{T}$ for some polynomials $p,p'$.
	First, notice that $|M_{\alpha}|$ is linear in the size of $\alpha$.
	This can easily be shown by a structural induction, where each construction adds a constant number of states to its subautomata only.

	Basic constructions $\mathcal{A}_{a_{\pi}}$ and $\mathcal{A}_{\lnot a_{\pi}}$ have constant size.
	For boolean connectives as well as $\langle \alpha \rangle \varphi, [\alpha]\varphi$ and $\Delta \alpha$, the construction of $\mathcal{A}_{\varphi}$ again just adds a constant number of states to the automata for the subformulas.
	The construction for $\lnot \Delta \alpha$ introduces a quadratic increase by \autoref{complementation}.
	Throughout the whole construction, this results in an exponential increase in the nesting depth of negated $\Delta \alpha$ constructs at most.
	More precisely, when bounding this nesting depth to a constant $d$, the polynomial $p$ on top of the exponential tower has degree at most $2d$.
	Existential quantifiers increase the size of the automaton exponentially in the size of the formula $\varphi$ and add a factor polynomial in the size of the structure $\mathcal{T}$.
	Using logarithmic laws, this translates to the form above.
	Note that the factor depending on $|\mathcal{T}|$ is added after the exponential blowup from the dealternation construction $\textit{MH}(\mathcal{A})$.
	
	It remains to show that the dealternation construction $\textit{MH}(\mathcal{A})$ introduces an exponential blowup of the structure's size at most once for formulas of criticality $0$, regardless of how many quantifiers the formula contains.
	In order to do this, we have to look closer at the proof of \autoref{dealternation}.
	We show that once the dealternation construction $\textit{MH}(\mathcal{A})$ is done for the innermost quantifiers, at most one state of each dealternised automaton has to be tracked in further dealternations.
	Thus, the exponential size of the subautomaton is added as a factor rather than in an exponent when determining the size of the state space of the full automaton.
	
	Our claim can be shown by an induction over the number of constructions on top of the dealternised automaton.
	In the base case, no construction is done on top of a dealternised automaton $\mathcal{A}_{\varphi}$.
	Since $\mathcal{A}_{\varphi}$ is a Büchi automaton, a run of the resulting automaton is a path rather than a tree on every word.
	Thus, only one state has to be tracked.
	In the inductive step, we discriminate cases for the outermost construction.
	By the induction hypothesis, at most one state of each dealternised automaton has to be tracked in each subautomaton.
	For the construction $\varphi_{1} \lor \varphi_{2}$, a run tree moving into $\mathcal{A}_{\varphi_1}$ or $\mathcal{A}_{\varphi_2}$ never returns to the initial state.
	Thus, since $\mathcal{A}_{\varphi_{1}}$ and $\mathcal{A}_{\varphi_2}$ are unconnected, we track states of only one of the automata.
	Then, the claim is implied by the induction hypothesis.
	The construction for $\mathcal{A}_{\varphi_{1} \land \varphi_2}$ works similarly, with the difference that we have to track states of both subautomata when a run moves into this automaton over the initial state.
	This does,  however, not lead to an increase in states of each dealternised automaton that have to be tracked, since these subautomata are unconnected.
	The next construction we have to consider is $\mathcal{A}_{\langle \alpha \rangle \varphi}$.
	Here, a run has the property that at most one state of $M_{\alpha}$ has to be tracked, which can be replaced by states of $\mathcal{A}_{\varphi}$ at some point.
	Additionally, arbitrarily many states of $\mathcal{A}_{\psi}$ for subformulas $\psi$ of $\alpha$ can be tracked.
	However, this is no contradiction to our claim, since $\alpha$ may not contain any quantified subformulas and thus $\mathcal{A}_{\psi}$ may not contain dealternised automata in uncritical formulas.
	Thus, since the induction hypothesis states that at most one state of each dealternised automaton of $\mathcal{A}_{\varphi}$ has to be tracked at any point, this shows our claim.
	As another case, we consider the construction for $\exists \pi . \varphi$.
	Here, since only a disjunctive transition is added on top of $\mathcal{A}_{\varphi}$, we can argue similar as in the case for $\varphi_1 \lor \varphi_2$ with the difference that we consider only a single subautomaton.
	Due to the exemption rule in the definition of criticality, there is an additional construction to be considered: $[\alpha] \varphi$, where $\alpha$ is deterministic and the outermost quantifier inside $\alpha$ is negated.\footnote{There are additional forms of $\alpha$, where explosion through a  negated quantifier inside the modality $[\alpha]$ can be avoided.
	This includes all forms where in any run in $M_{\alpha}$, a test for $\lnot \psi$ occurs only in a situation where all states occurring at the same level of the run can transition into $\mathcal{A}_{\lnot \psi}$.
	Then, when a transition into $\mathcal{A}_{\lnot \psi}$ can be taken in one of the states, it can be taken in all of the states.
	Since they are on the same level, the same continuation in $\mathcal{A}_{\lnot \psi}$ can be used for all these branches, such that only a single state of each dealternised subautomaton of $\mathcal{A}_{\lnot \psi}$ needs to be tracked.}
	Since tests $\psi$ in $\alpha$ are handled by disjunctively transitioning into the automaton for $\lnot \psi$ in the $[\alpha] \varphi$ construction, this cancels out the negation of the quantifier.
	Therefore, no critical negation construction has to be performed on a dealternised subautomaton during the construction of $\mathcal{A}_{\lnot \psi}$.
	Then, since due to the fact that $M_{\alpha}$ is deterministic, conjunctions of transitions in $M_{\alpha}$ behave the same as disjunctions and we can argue just as in the case for $\mathcal{A}_{\langle \alpha \rangle \varphi}$.
	Finally, observe that we do not have to consider constructions for $\Delta \alpha$, $\lnot \Delta \alpha$, or general $[\alpha] \varphi$, since the resulting formula is not uncritical when any of these contain a quantified subformula.

	\textbf{Inductive step: }
	$\varphi$ has criticality $k+1$.
	On the path in $\varphi$'s syntax tree inducing the criticality, we inspect the outermost critical quantifier.
	Its subformulas $\varphi_i$ have criticality at most $k$.
	Using the induction hypothesis on all subformulas, we obtain automata of size at most $\mathcal{O}(g(k+1,|\varphi_i|+log(|\mathcal{T}|)))$.
	The next dealternation will result in an automaton of size $2^{\mathcal{O}(|\mathcal{A}_{\varphi_i}|)}$ (by \autoref{dealternation}) which can be bounded by $2^{\mathcal{O}(g(k+1,|\varphi_i|+log(|\mathcal{T}|)))} = \mathcal{O}(g(k+2,|\varphi|+log(|\mathcal{T}|)))$.
	Since we inspected the outermost critical quantifier on the path inducing the criticality of the formula, any of the subsequent constructions will not cause a further exponential blowup of the automaton's size, as argued in the base case.
\end{proof}

%% file: math/proofs/hardness.tex
\begin{proof}[Proof of \autoref{hardness}]	
	We reduce the model checking problem of alternation depth $k$ HyperCTL*-formulas, which by \cite{Clarkson2014} is hard for NSPACE$(g(k,|\varphi|))$ and for NSPACE$(g(k-1,|\mathcal{T}|))$, to the criticality $k$ model checking problem in our logic.
	For this, we translate a given Kripke structure $K$ and a HyperCTL* formula $\varphi$ with alternation depth $k$ to a KTS $tr(K)$ and a HyperPDL-$\Delta$ formula $tr(\varphi)$ of criticality $k$ such that $K \models \varphi$ iff $tr(K) \models tr(\varphi)$.
	
	
	The translation of the Kripke structure $K = (S, s_0, \delta, L)$ is the structure itself, seen as a KTS:
	We keep $S,s_0, L$ and define $\delta_{\sigma} := \delta$ for $\Sigma = \{\sigma\}$.
	If one wants to define the KTS over a larger set $\Sigma$, this would be possible by choosing the same $\delta_{\sigma}$ for all $\sigma \in \Sigma$.
	This would however unnecessarily complicate this proof, since the set of transitions that can be taken is not affected by the choice of atomic programs $\sigma$ in any state $s \in S$.
	
	
	For formulas $\varphi$, only the cases $\bigcirc \varphi$, $\varphi_1 \mathcal{U} \varphi_2$ and $\varphi_1 \mathcal{R} \varphi_2$ are non-trivial.
	We translate them in the following way: $tr(\bigcirc \varphi) := \langle \bullet \rangle tr(\varphi)$ or $tr(\bigcirc \varphi) := [ \bullet ] tr(\varphi)$, such that $\langle \bullet \rangle tr(\varphi)$ is used in the negation normal form (this way a translation of $\bigcirc \varphi$ never induces criticality on its own).
	Furthermore, we translate $tr(\varphi_1 \mathcal{U} \varphi_2) := \langle (tr(\varphi_1) ? \bullet)^{*} \rangle tr(\varphi_2)$ and $tr(\varphi_1 \mathcal{R} \varphi_2) := [(tr(\lnot \varphi_{1}) ? \bullet)^{*}] tr(\varphi_2)$.
	Since existential and universal quantifiers are not changed, $tr(\bigcirc \varphi)$ does not contain a $\psi ?$ construct, and $\mathcal{U}$ (resp. $\mathcal{R}$) modalities are translated to a single use of a $\langle \alpha \rangle \varphi$ (resp. $[\alpha]$) construct with use of $\psi ?$ where $\varphi_i$ and $tr(\varphi_i)$ have the same role in the definition of alternation depth and criticality, $\varphi$ and $tr(\varphi)$ correspond in alternation depth resp. criticality.
	Note, that in the case of $\mathcal{R}$ this only holds due to the exemption of negated outermost quantifiers in tests occuring in modalities $[\alpha]$ for deterministic $\alpha$.
	
	
	We now show that $K \models \varphi$ iff $tr(K) \models tr(\varphi)$ indeed holds.
	First, notice that a path assignment $\Pi$ on $K$ with $\Pi(\pi_i) = s_0^i s_1^i ...$ can be directly translated to a path assignment $tr(\Pi)$ on $tr(K)$ with $tr(\Pi)(\pi_i) = s_0^i \sigma s_1^i \sigma ...$.
	Then, the assumption immediately follows from the statement $\Pi[i,\infty] \models_{K} \varphi$ iff $tr(\Pi)[2i, \infty] \models_{tr(K)} tr(\varphi)$ which can be shown by structural induction on $\varphi$.
	In all cases where the translation $tr$ is trivial, the proof is as well.
	Let us consider the non-trivial cases:
	
	$\bigcirc \varphi$: Note that for $\alpha = \bullet = (\cdot,...,\cdot)$, $(\Pi,0,2) \in R(\alpha)$ trivially holds for any path assignment $\Pi$.
	Thus, checking $tr(\Pi) \models_{tr(K)} tr(\bigcirc \varphi)$ reduces to checking whether $tr(\Pi)[2,\infty] \models_{tr(K)} tr(\varphi)$ holds.
	Then, the assumption immediately follows from the induction hypothesis.
	
	$\varphi_1 \mathcal{U} \varphi_2$: Using our semantics definition and our above argument about wildcard programs in $\alpha$, the model checking problem for $\langle (tr(\varphi_1) ? \bullet)^{*} \rangle tr(\varphi_2)$ translates to 
	\begin{align*}
		&\exists i \geq 0. tr(\Pi)[i,\infty] \models_{tr(K)} tr(\varphi_{2}) \land \\
		&\exists n \geq 0, 0 = k_1 \leq k_2 \leq ... \leq k_n = i. \\
		&\forall 1 \leq l < n . k_{i+1} = k_i + 2 \land tr(\Pi)[k_l,\infty] \models_{tr(K)} tr(\varphi_1)
	\end{align*}
	In this condition, notice that both $i$ and all $k_l$ are even.
	Together with the structure of $tr(\Pi)$, this enables us to use the induction hypothesis on $tr(\Pi)[i,\infty] \models_{tr(K)} tr(\varphi_{2})$ and $tr(\Pi)[k_l,\infty] \models_{tr(K)} tr(\varphi_1)$ to show that the above translation is equivalent to the semantics of $\mathcal{U}$ on $\Pi$, which shows the assumption.
	
	$\varphi_1 \mathcal{R} \varphi_2$: similar to case $\varphi_1 \mathcal{U} \varphi_2$
\end{proof}

%% file: math/proofs/qptltranslation.tex
\begin{proof}[Proof of \autoref{qptl}, encoding into \text{MSO[E]}]
We construct an inductive translation, using the induction hypothesis that (i) for every program $\alpha$, there is an MSO[E] formula $\psi_{\alpha}(\Pi, v_i, v_j)$ equivalent to $(\Pi, i, j) \in R(\alpha)$ when $i$ and $j$ are the indices of $v_i$ and $v_j$ on $\pi_1$ in $\Pi$ and (ii) for every HyperPDL-$\Delta$ formula $\varphi$, there is an equivalent MSO[E] formula $tr(\varphi$).
Atomic propositions, logical conjunctions and disjunctions and negations are trivial and quantified paths can be translated exactly as in the translation of HyperQCTL$^*$ to MSO[E] in \cite{Coenen2019}.
We therefore only consider regular expressions.
These can be translated analogous to Büchi's classical proof of the equivalence of MSO and regular expressions on finite words with the only difference being that multiple paths (as sets) are considered simultaneously and synchronised with the $E$ predicate.
 We describe the case of $\Delta \alpha$ only and for all other cases, similar modifications of the Büchi approach can be used.
$\Delta \alpha$ is translated as 
\begin{align*}
\exists X_{\pi_1} \dots \exists X_{\pi_n}. &\bigwedge_{1 \leq i \leq n} (X_{\pi_i}  \subseteq \pi_i \land E(min(X_{\pi_i}), min(X_{\pi_1})) \\
&\land \forall k \in X_{\pi_1}. \forall l \in X_{\pi_i} . E(k, l) \rightarrow E(next(X_{\pi_1}, k),next(X_{\pi_i}, l))) \\
&\land \forall v_1 \in X_{\pi_1} \dots \forall v_n \in X_{\pi_n}.  \\
&\qquad(\bigwedge_{1 \leq i \leq n} E(v_1, v_i)) \rightarrow \psi_{\alpha}(\Pi, v_1, next(X_{\pi_1}, v_1)).
\end{align*}
Here, $min(X)$ (denoting the minimal element of $X$) and $next(X, v)$ (denoting the next element in $X$ after $v$) are used in a functional notation, but can obviously be converted to a purely predicative notation.
Intuitively, the sets $X_{\pi_1},...,X_{\pi_n}$ are the subsets of nodes from $\pi_1,...,\pi_n$, where indices $k_1,k_2,...$ from the definition of semantics are matched.
\end{proof}

%% file: math/constructions/alphaalt.tex

In the main body of the paper, we have argued that the transition function in $\mathcal{A}_{\varphi}$ for $\varphi$ containing $\alpha$ can be exponential in the size of $\alpha$.
Consider the automaton in \autoref{fig:exponentialtransition} where states annotated with $\psi_{i,j}$ are marked with the corresponding formula and each unannotated edge stands for an $\varepsilon$-transition.
Such an automaton can occur when $\alpha$ has the form $(\psi_{0,1}? + \psi_{0,2}?)(\psi_{1,1}? + \psi_{1,2}?)...(\psi_{n,1}? + \psi_{n,2}?)\tau$ and should serve as an illustrative example.
We consider the case where this $\alpha$ is used in a $\langle . \rangle$ formula.
For ease of presentation, we assume that $\psi_{i,j}$ can be tested by moving into a state $p_{i,j}$.
It is possible to reach $q_f$ from $q_0$ with an exponential number of $\varepsilon$-paths, each with a different combination of markings.
Therefore we have $q_0 \xRightarrow{\tau}_X q$ for exponentially many $X$, transferring into the size of the transition function when constructing $\rho(q_0,\tau) \equiv \bigvee \{q \land \bigwedge_{\psi_{i,j} \in X} p_{i,j} \mid q_0 \xRightarrow{\tau}_X q\}$.

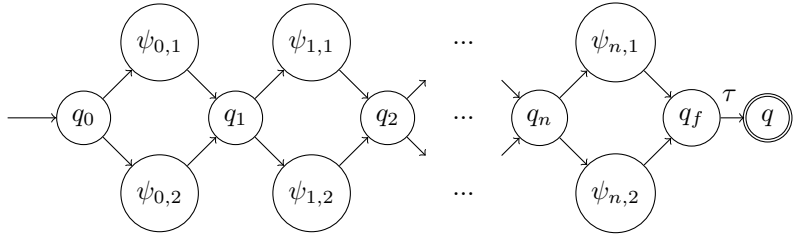
\begin{figure}
	\centering
	\begin{tikzpicture}
		\node[draw,circle] at (-4,0) (q0){$q_0$};
		\node[draw,circle] at (-3,1) (q01){$\psi_{0,1}$};
		\node[draw,circle] at (-3,-1) (q02){$\psi_{0,2}$};
		\node[draw,circle] at (-2,0) (q1){$q_1$};
		\node[draw,circle] at (-1,1) (q11){$\psi_{1,1}$};
		\node[draw,circle] at (-1,-1) (q12){$\psi_{1,2}$};
		\node[draw,circle] at (0,0) (q2){$q_2$};
		
		\node[draw,circle] at (2,0) (qn){$q_n$};
		\node[draw,circle] at (3,1) (qn1){$\psi_{n,1}$};
		\node[draw,circle] at (3,-1) (qn2){$\psi_{n,2}$};
		\node[draw,circle] at (4,0) (qf){$q_f$};
		\node[draw,double,circle] at (5,0) (q){$q$};
		
		\node[] at (1,1) {...};
		\node[] at (1,0) {...};
		\node[] at (1,-1) {...};
		\node[] at (4.5,0.3) {$\tau$};
		
		\path[->] (-5,0) edge (q0);
		\path[->] (q0) edge (q01);
		\path[->] (q0) edge (q02);
		\path[->] (q01) edge (q1);
		\path[->] (q02) edge (q1);
		\path[->] (q1) edge (q11);
		\path[->] (q1) edge (q12);
		\path[->] (q11) edge (q2);
		\path[->] (q12) edge (q2);
		\path[->] (q2) edge (0.5,0.5);
		\path[->] (q2) edge (0.5,-0.5);
		\path[->] (1.5,0.5) edge (qn);
		\path[->] (1.5,-0.5) edge (qn);
		\path[->] (qn) edge (qn1);
		\path[->] (qn) edge (qn2);
		\path[->] (qn1) edge (qf);
		\path[->] (qn2) edge (qf);
		\path[->] (qf) edge (q);
	\end{tikzpicture}
	\caption{Automaton $M_{\alpha}$ with an exponential number of test combinations}
	\label{fig:exponentialtransition}
\end{figure}

For this example, it is easy to see that these exponentially many different combinations of transitions could equally be represented by a formula of a much smaller size, namely $(p_{0,1} \lor p_{0,2}) \land (p_{1,1} \lor p_{1,2}) \land ... \land (p_{n,1} \lor p_{n,2})$.
This is due to the fact that conjunction and disjunction closely resemble the behaviour of concatenation and sum constructions in $M_{\alpha}$ when considering $\varepsilon$-paths for the construction of $\rho$.
We will show here, that using these ideas it is possible to construct such a formula of size not greater than $3 \cdot |\alpha| + 2$ for every $\alpha$.

We proceed by constructing a function $\varepsilon p$ such that $\varepsilon p (q,q',\alpha) = \vartheta$ for a formula $\vartheta$ equivalent to $\bigvee \{ \bigwedge_{\psi_i \in X} v_i \mid q \xRightarrow{\varepsilon}_X q'\}$ for $\xRightarrow{\varepsilon}_X$ constructed from $M_{\alpha}$.
Here we use variables $v_i$ as placeholders for formulas $\psi_i$ to be later replaced by a transition into the corresponding automaton $\mathcal{A}_{\psi_i}$.
For $\tau$-transitions, this can straightforwardly be extended to a function $\tau p$ which can then be used to construct the $\tau$ transition function for a state $q$ in a more succinct way.
For a $\langle \alpha \rangle \varphi$ formula and some $q \in Q_{\alpha}$ we then have $\rho(q,(\mathpzc{s},\tau)) = \bigvee \{q' \land \tau p (q,q',\alpha)[\rho_{\psi_i}(q_{0,\psi_i},(\mathpzc{s},\tau)) / v_i] \mid q' \in Q\}$.
Some remarks are in order for this transition function construction: (i) in this construction, opposed to the one used before, each $q' \in Q$ can occur at most once in the disjunction, (ii) for $q' \in Q$ such that there is no $X$ with $q \xRightarrow{\tau}_X q'$, i.e. $q'$ is not reachable from $q$ with $\tau$, we have $\tau p (q,q',\alpha) \equiv false$ and thus the state can be eliminated from the disjunction and (iii) for $[.]$ formulas, $\tau p$ and $\rho$ can similarly be constructed in a dual way.

Since there is no direct way to create the desired formula for $\alpha = \alpha_1^*$ while meeting the size constraints, we cannot do a direct inductive construction for $\varepsilon p$.
Instead we perform an inductive construction $dp$ that is similar to $\varepsilon p$ but does not take \textit{backwards edges} $(q_f,q_0)$ originating from $\alpha^*$-constructions into account.
Then, $\varepsilon p$ can be constructed from $dp$ by only considering a single backwards edge for each pair of states.
The idea behind this is that each path $p_*$ considering more than one backwards edge is \textit{subsumed} by some path $p_1$ considering only one backwards edge in the sense that if $p_*$ visits the set $X_*$ of markings and $p_1$ visits the set $X_1$ of markings, then $X_1 \subseteq X_*$.
Then, the conjunction over $X_*$ is implied by the conjunction over $X_1$.
Since in the construction of $\rho$, we perform a disjunction over all paths with a conjunction over all seen markings inside, we can then omit $p_*$ from the disjunction.

\textbf{Construction of $dp$ and $\varepsilon p$:}
First, we construct $dp$ inductively.

\begin{align*}
	\alpha = \tau \quad& dp(q,q',\alpha) = \begin{cases}
	\false &\text{if } q \neq q' \\
	\true &\text{else}
	\end{cases} \\
	\alpha = \varepsilon \quad& dp(q,q',\alpha) = \begin{cases}
	\false &\text{if } q = q_1 \text{ and } q' = q_0 \\
	\true &\text{else}
	\end{cases} \\
	\alpha = \alpha_1 + \alpha_2 \quad& dp(q,q',\alpha) = \begin{cases}
	dp (q,q',\alpha_i) &\text{if } q,q' \in Q_i \\
	dp (q_{0,i},q',\alpha_i) &\text{if } q = q_0 \text{ and } q' \in Q_i \\
	dp (q,q_{f,i},\alpha_i) &\text{if } q \in Q_i \text{ and } q'= q_f \\
	dp (q_{0,1},q_{f,1},\alpha_1) \lor dp (q_{0,2},q_{f,2},\alpha_2) &\text{if } q = q_0 \text{ and } q' = q_f \\
	\false &\text{else}
	\end{cases} \\
	\alpha = \alpha_1 \cdot \alpha_2 \quad& dp(q,q',\alpha) = \begin{cases}
	dp (q,q',\alpha_i) &\text{if } q,q' \in Q_i \\
	dp (q,q_{f,1},\alpha_1) \land dp (q_{0,2},q',\alpha_2)&\text{if } q \in Q_1 \text{ and } q' \in Q_2 \\
	\false &\text{else}
	\end{cases} \\
	\alpha = (\alpha_1)^* \quad& dp(q,q',\alpha) = \begin{cases}
	dp (q,q',\alpha_1) &\text{if } q,q' \in Q_1 \\
	\true &\text{if } q = q_0 \text{ and } q' = q_f \\
	dp (q_{0,1},q',\alpha_1) &\text{if } q = q_0 \text{ and } q' \in Q_1 \\
	dp (q,q_{f,1},\alpha_1) &\text{if } q \in Q_1 \text{ and } q' = q_f \\
	\false &\text{else}
	\end{cases} \\
	\alpha = \psi_k? \quad& dp(q,q',\alpha) = \begin{cases}
	\true &\text{if } q = q' \neq q_1 \\
	\false &\text{if } q = q_i, q' = q_j, i > j \\
	v_k &\text{else}
	\end{cases}
\end{align*}
Using $dp$, we are now able to construct $\varepsilon p$ directly for all $q,q'$ and $\alpha$.
\begin{align*}
	\varepsilon p (q,q',\alpha) = dp(q,q',\alpha) \lor (dp(q,q_{f,\bar{\alpha}},\alpha) \land dp(q_{0,\bar{\alpha}},q',\alpha))
\end{align*}
Here $\bar{\alpha}$ is the innermost $*$-construction that contains both $q$ and $q'$.
In case no such $\bar{\alpha}$ exists, both $dp(q,q_{f,\bar{\alpha}},\alpha)$ and $dp(q_{0,\bar{\alpha}},q',\alpha)$ are given by $\false$ instead.

\textbf{Theoretical justification:}

In order to use this succinct alternative in our construction, we have to argue that it indeed has the desired properties.
Therefore we establish a number of theorems:

\input{math/theorems/altconstructionsize}

\input{math/theorems/altconstructionlemma}

\input{math/theorems/altconstruction}

%% file: math/theorems/altconstructionsize.tex
\begin{theorem}
	$|dp (q,q',\alpha)| \leq |\alpha|$ and $|\varepsilon p (q,q',\alpha)| \leq 3 \cdot |\alpha| + 2$ for all $q,q'$ and $\alpha$.
\end{theorem}

\begin{proof}
	The first claim can be established by a straightforward structural induction on $\alpha$.
	It is easy to see that in each case of the construction, at most one operator is added to $dp (q,q',\alpha)$ and each partial term is used at most once.
	
	The second claim follows directly from the first claim and the definition of $\varepsilon p$.
\end{proof}

%% file: math/theorems/altconstructionlemma.tex
\begin{lemma}\label{altconstructionlemma}
	We have $dp(q,q',\alpha) \equiv \bigvee \{\bigwedge_{\psi_i \in X} v_i \mid q \xRightarrow{\varepsilon}_X q'\}$ for $\xRightarrow{\varepsilon}_X$ constructed from $M_{\alpha}$ where all backwards edges from $*$-constructions are removed.
\end{lemma}

\begin{proof}
	We show this claim by a structural induction on $\alpha$.
	
	\textbf{Case \boldmath$\alpha = \tau$\unboldmath:}
	There are two unmarked states $q_0,q_1$ in $M_{\alpha}$ with a $\tau$-transition connecting them.
	There are no $\varepsilon$-transitions.
	Thus, we have $q \xRightarrow{\varepsilon}_X q'$ iff $q = q'$ and $X = \emptyset$.
	Therefore we have $\bigvee \{\bigwedge_{\psi_i \in X} v_i \mid q \xRightarrow{\varepsilon}_X q'\} \equiv \false$ for $q \neq q'$ and $\bigvee \{\bigwedge_{\psi_i \in X} v_i \mid q \xRightarrow{\varepsilon}_X q'\} \equiv \true$ for $q = q'$, establishing the claim.
	
	\textbf{Case \boldmath$\alpha = \varepsilon$\unboldmath:}
	There are two unmarked states $q_0,q_1$ in $M_{\alpha}$ with an $\varepsilon$-transition connecting $q_0$ to $q_1$.
	Thus, we have $q \xRightarrow{\varepsilon}_X q'$ iff $X = \emptyset$ and either $q \neq q_1$ or $q' \neq q_0$.
	Therefore we have $\bigvee \{\bigwedge_{\psi_i \in X} v_i \mid q \xRightarrow{\varepsilon}_X q'\} \equiv \false$ for $q = q_1$,$q' = q_0$ and $\bigvee \{\bigwedge_{\psi_i \in X} v_i \mid q \xRightarrow{\varepsilon}_X q'\} \equiv \true$ else, establishing the claim.
	
	\textbf{Case \boldmath$\alpha = \alpha_1 + \alpha_2$\unboldmath:}
	By induction hypothesis, the claim holds for $\alpha_1$ and $\alpha_2$.
	To obtain $M_{\alpha}$ from $M_{\alpha_1}$ and $M_{\alpha_2}$, a new starting and final state are added with $\varepsilon$-transitions to the old starting states and from the old final states, respectively.
	We consider different cases how a path inducing $q \xRightarrow{\varepsilon}_X q'$ could have been constructed.
	In the first case, where both $q$ and $q'$ are in the same automaton $M_{\alpha_i}$, no additional paths could have been introduced by the new transitions.
	Thus, the claim follows immediately from the induction hypothesis.
	In the second case, where $q = q_0$ and $q'$ is in $M_{\alpha_i}$ a path must take the $\varepsilon$-transition to $q_{0,i}$ and then take a path between $q_{0,i}$ and $q'$.
	Since $q_0$ is not marked, we have $q_0 \xRightarrow{\varepsilon}_X q'$ iff $q_{0,i} \xRightarrow{\varepsilon}_X q'$, establishing the claim by induction hypothesis.
	The third case, where $q$ is in $M_{\alpha_i}$ and $q' = q_f$ is analogous to the second one.
	Another case is $q = q_0$ and $q' = q_f$.
	Since $M_{\alpha_1}$ and $M_{\alpha_2}$ are not connected, we have $q \xRightarrow{\varepsilon}_X q'$ iff $q_{0,1} \xRightarrow{\varepsilon}_X q_{f,1}$ or $q_{0,2} \xRightarrow{\varepsilon}_X q_{f,2}$ with the same argument as used in cases two and three.
	Since no paths are added from $q_{0,i}$ to $q_{f,i}$ when going over from $M_{\alpha_i}$ to $M_{\alpha}$, $q_{0,i} \xRightarrow{\varepsilon}_X q_{f,i}$ holds for $\xRightarrow{\varepsilon}_X$ constructed from $M_{\alpha_i}$ iff it holds for $\xRightarrow{\varepsilon}_X$ constructed from $M_{\alpha}$.
	Therefore we have $\bigvee \{\bigwedge_{\psi_i \in X} v_i \mid q \xRightarrow{\varepsilon}_X q' \} \equiv \bigvee \{\bigwedge_{\psi_i \in X} v_i \mid q_{0,1} \xRightarrow{\varepsilon}_X q_{f,1} \} \lor \bigvee \{\bigwedge_{\psi_i \in X} v_i | q_{0,2} \xRightarrow{\varepsilon}_X q_{f,2} \} \equiv dp (q_{0,1},q_{f,1},\alpha_1) \lor dp (q_{0,2},q_{f,2},\alpha_2) = dp (q,q',\alpha)$ using the induction hypothesis.
	The remaining cases include $q = q_f$ with $q' = q_0$ and $q$ being in $M_{\alpha_i}$ with $q'$ being in $M_{\alpha_{1-i}}$.
	In both cases, $q'$ is not reachable from $q$ using only $\varepsilon$-transitions.
	Thus, the claim is established with similar arguments as in previous cases.
	
	\textbf{Case \boldmath$\alpha = \alpha_1 \cdot \alpha_2$\unboldmath:}
	We consider three cases.
	In the first one, $q$ and $q'$ are both in $M_{\alpha_i}$.
	Since only transitions from $M_{\alpha_1}$ into $M_{\alpha_2}$ are possible but not backwards, a path inducing $q \xRightarrow{\varepsilon}_X q'$ has to stay in $M_{\alpha_i}$ the whole time.
	The claim then follows from the induction hypothesis.
	In the second case, $q$ is in $M_{\alpha_1}$ and $q'$ is in $M_{\alpha_2}$.
	A path from $q$ to $q'$ has to transition through $q_{f,1}$ and $q_{0,2}$ to be able to switch automata, thus we have $q \xRightarrow{\varepsilon}_X q'$ iff $q \xRightarrow{\varepsilon}_Y q_{f,1}$ and $q_{0,2} \xRightarrow{\varepsilon}_Z q'$ for some $Y,Z$ with $X = Y \cup Z$.
	Since for state pairs inside one of the subautomata it does not matter whether $\xRightarrow{\varepsilon}_X$ was constructed from $M_{\alpha}$ or $M_{\alpha_i}$, the claim follows from the induction hypothesis.
	In the last case, $q$ is in $M_{\alpha_2}$ and $q'$ is in $M_{\alpha_1}$.
	Since $q'$ is not reachable from $q$, $q \xRightarrow{\varepsilon}_X q'$ can not hold for any $X$ and the claim follows immediately.
	
	\textbf{Case \boldmath$\alpha = \alpha_1^*$\unboldmath:}
	We consider five cases.
	In the first case, we have $q,q' \in Q_1$.
	Since the backwards edge that was added during the construction is ignored for this lemma, no new paths from $q$ to $q'$ are added compared to $M_{\alpha_1}$.
	Therefore the claim follows from the induction hypothesis.
	In the second case, we have $q = q_0$ and $q' = q_f$.
	The claim follows immediately from the fact that there is an $\varepsilon$-edge in between the two states.
	The third case, where $q = q_0$ and $q' \in Q_1$, and the fourth case, where $q \in Q_1$ and $q' = q_f$ work in a similar way by considering the added $\varepsilon$ edges between old and new starting and final states and by using the induction hypothesis.
	In the last case, we have $q = q_f$ and $q' = q_0$.
	The claim holds since the backwards edge connecting the two states is ignored for this lemma.
	
	\textbf{Case \boldmath$\alpha = \psi_k ?$\unboldmath:}
	We consider the different cases how $q'$ can be reached by  $\varepsilon$-transitions from $q$ in $M_{\alpha}$.
	In the first case, $q = q'$ with $q \neq q_1$, we have trivial reachability without encountering a state marking.
	Here, the claim is established as in previous cases.
	In the second case, where $q = q_i, q' = q_j$ with $i > j$, $q'$ is not reachable from $q$ since the $\varepsilon$-transitions only point in the other direction.
	The claim is again established as in previous cases.
	In all other cases, $q'$ can be reached from $q$ with $\varepsilon$-transitions, but only with encountering the state marking in $q_1$.
	Since this is the only state marking in $M_{\alpha}$, we have $\bigvee \{\bigwedge_{\psi_i \in X} v_i \mid q \xRightarrow{\varepsilon}_X q'\} \equiv v_i = dp (q,q',\alpha)$, establishing the claim.
\end{proof}

%% file: math/theorems/altconstruction.tex
\begin{theorem}
	We have $\varepsilon p(q,q',\alpha) \equiv \bigvee \{\bigwedge_{\psi_i \in X} v_i \mid q \xRightarrow{\varepsilon}_X q'\}$ for $\xRightarrow{\varepsilon}_X$ constructed from the automaton $M_{\alpha}$.
\end{theorem}

\begin{proof}
	Compared to the claim made about $dp$ in \autoref{altconstructionlemma}, backwards edges must now be considered in our claim about $\varepsilon p$.
	The central observation is that the contribution of all $\varepsilon$-paths from $q$ to $q'$ is already captured by two particular types of $\varepsilon$-paths: either by going from $q$ to $q'$ directly without taking a backwards edge, or by taking only the backwards edge from the construction of $M_{\bar{\alpha}}$ exactly once (where $\bar{\alpha}$ is the innermost *-construction that contains both $q$ and $q'$).
	As was shown in \autoref{altconstructionlemma}, the first type is captured by $dp(q,q',\alpha)$.
	It is also straightforward to see from \autoref{altconstructionlemma} that the second type is captured by $dp(q,q_{f,\bar{\alpha}},\alpha) \land dp(q_{0,\bar{\alpha}},q',\alpha)$.
	
	We now argue that  further backwards edges need not be considered and thus all paths are subsumed by these two cases.
	In order to use a backwards edge outside of $M_{\bar{\alpha}}$, a path has to leave $M_{\bar{\alpha}}$ via $q_{f,\bar{\alpha}}$ and finally reenter it via $q_{0,\bar{\alpha}}$.
	The contribution of such paths to the disjunction is subsumed by paths taking the backwards edge from $q_{f,\bar{\alpha}}$ to $q_{0,\bar{\alpha}}$ directly.
	Backwards edges on the paths from $q$ to $q_{f,\bar{\alpha}}$ or on the paths from $q_{0,\bar{\alpha}}$ to $q'$ on the other hand that originate in a final state $q_f$ of some subautomaton only lead to paths that later visit $q_f$ a second time. Thus their contribution is again subsumed by the contribution of the path where loops from $q_f$ to itself are cut out.
\end{proof}